\documentclass[aap,preprint]{imsart}
\usepackage{graphicx}
\usepackage{amsmath,amssymb,amsthm}

\usepackage[OT1]{fontenc}
\usepackage[numbers]{natbib}
\usepackage[colorlinks,citecolor=blue,urlcolor=blue]{hyperref}

\startlocaldefs \numberwithin{equation}{section} \theoremstyle{plain}
\newtheorem{thm}{Theorem}[section] \newtheorem{lemma}{Lemma}[section]

\newtheorem{proposition}{Proposition}[section]

\endlocaldefs

\usepackage{a4wide}

\begin{document}

\newcommand{\gt}{\rightarrow}
\newcommand{\veps}{\varepsilon}
\newcommand{\E}{\mathop{\mathbb{E}}}
\newcommand{\cov}{\mathop{\text{cov}}}
\newcommand{\lb}{\left( }
\newcommand{\rb}{\right) }%
\begin{frontmatter}
  \title{On singular value distribution of large-dimensional
    autocovariance matrices}
  \runtitle{Singular values of  autocovariance matrices}

%  \begin{aug}
%    \author{\fnms{Zeng} \snm{Li}\thanksref{t1}\ead[label=e1]{u3001205@hku.hk}},
%    \author{\fnms{Guangming} \snm{Pan}\thanksref{t2}
%      \ead[label=e2]{gmpan@ntu.edu.sg}}
%    \and
%    \author{\fnms{Jianfeng}
%      \snm{Yao}\thanksref{t3}\ead[label=e3]{jeffyao@hku.hk}}
%   \thankstext{t1}{Corresponding author.}
%    \thankstext{t2}{Guangming Pan's work is partly supported by
%      Grant xxxx.}
%   \thankstext{t3}{Jianfeng Yao       acknowledge the support from
%      a RGC Grant xxxx.}
%
%    \runauthor{Z. Li, G. Pan and J. Yao}
%
%    \affiliation{The University of Hong Kong\thanksref{t1},
%      and
%      The Nanyang Technological University\thanksref{t2}}
%
%    \address{Department of Statistics and Actuarial Science\thanksref{t1,t3}\\
%      The University of Hong Kong\\
%      \printead{e1,e3}}
%
%    \address{School of Physical \& Mathematical Sciences\thanksref{t2}\\
%      Nanyang Technological University\\
%      \printead{e2}}
%  \end{aug}

\begin{aug}
\author{\fnms{Zeng} \snm{Li}\thanksref{t1,m1}\ead[label=e1]{u3001205@hku.hk}},
\author{\fnms{Guangming} \snm{Pan}\thanksref{m2}\ead[label=e2]{gmpan@ntu.edu.sg}}
\and
\author{\fnms{Jianfeng} \snm{Yao}\thanksref{m1}
\ead[label=e3]{jeffyao@hku.hk}}
\thankstext{t1}{Corresponding author.}
\runauthor{Z. Li, G. Pan and J. Yao}

\affiliation{The University of Hong Kong\thanksmark{m1} and The Nanyang Technological University\thanksmark{m2}}

\address{Department of Statistics and Actuarial Science\\
      The University of Hong Kong\\
      \printead{e1,e3}}
 \address{School of Physical \& Mathematical Sciences\\
      Nanyang Technological University\\
      \printead{e2}}
\end{aug}

  \begin{abstract}
   Let $(\veps_j)_{j\geq 0}$ be a sequence of independent $p-$dimensional random vectors and $\tau\geq1$ a given integer.
   From a sample $\veps_1,\cdots,\veps_{T+\tau-1},\veps_{T+\tau}$ of the sequence, the so-called lag$-\tau$ auto-covariance matrix is
   $C_{\tau}=T^{-1}\sum_{j=1}^T\veps_{\tau+j}\veps_{j}^t$.
    When the dimension $p$ is large compared to the sample size $T$,
    this paper establishes the limit of the singular value
    distribution of $C_\tau$  assuming that $p$ and $T$ grow to
    infinity proportionally and the sequence satisfies  a Lindeberg
    condition. %on fourth order moments.
   Compared to existing asymptotic results on sample covariance matrices developed in random matrix theory,
    the case of an auto-covariance matrix
    is much more involved due to the fact that the summands are dependent and the matrix $C_\tau$ is not symmetric.
    Several new techniques are introduced for
    the derivation of the main theorem.
  \end{abstract}



\end{frontmatter}

\section{Introduction}

Let  $\veps_1,\ldots,\veps_{T+\tau}$ be a sample
from a  stationary process with values in $\mathbb{R}^p$. The $p\times p$ matrix
\begin{equation}\label{Ceps}
  C_\tau := \frac1T \sum_{j=1}^T \veps_{\tau+j} \veps_{j}^t,
\end{equation}
is the so-called lag$-\tau$ {\em sample auto-covariance matrix}
of the process (here $u^t$ denotes the transpose of a vector or
matrix $u$).
  In a  classical low-dimensional situation
where the dimension $p$ is assumed much smaller than the sample size $T$,
$C_\tau$ is very close to $\E C_\tau=\E\veps_{1+\tau} \veps_{1}^t$ so that its
asymptotic behavior when $T\to\infty$ (so $p$ is considered as fixed)
is well known. In the high-dimensional context where
typically  the dimension $p$ is of same order as $T$,
$C_\tau$  will  not converge to $\E C_\tau$  and
its asymptotic properties have not been  well investigated.
In this paper, we study the empirical spectral distribution (ESD)
of $C_\tau$, namely, the
distribution generated by its $p$ singular
values. The main result of the paper is the establishment of
the limit of this ESD when $(\veps_j)$ is an independent sequence with elements having a
finite fourth moments while
$p$ and $T$ grow to infinity proportionally.

In order to understand the importance of limiting spectral
distribution (LSD)
of singular values  of the auto-covariance matrix
$C_\tau$,  we describe a statistical problem where these
distributions are of central interest.
In a recent stimulating paper,
\citet{LamYao12} considers the following dynamic factor model
\begin{equation}\label{model}
  {x}_i=\Lambda {f}_i + \veps_i +\mu,
\end{equation}
where $\{{x}_i; ~ 0\le i\le T\}$
is an observed $p$-dimensional sequence,
$\{f_i\}$  a sequence of $m$-dimensional ``latent factor" ($m\ll p$)
uncorrelated with the error process $\{\veps_i\}$ and $\mu\in \mathbb{R}^p$ is the general mean.
%Notice that   nothing on  the right-hand side of the  equation
%is observable and the inference problem is challenging.
A particularly important question  here is the  determination of
the number $m$ of factors.
For any stationary process $\{  w_i\}$, let
 ${\Sigma}_{  w} =\cov({  w}_i,{  w}_{i-1})$  be its
(population) lag-$1$  auto-covariance matrix, we have
\[     {\Sigma}_{x} = \Lambda {\Sigma}_{f} \Lambda^t.
\]
It turns out that ${\Sigma}_{x}$ has exactly $m$ non-null singular
values so that based on a sample ${x}_0,{x}_1,\ldots,{x}_T$
it seems natural  to infer $m$ from
the  singular values of the sample lag-1 auto-covariance matrix
%\[  C_{x} = \frac 1T \sum_{j=1}^T {x}_j {x}^t_{j-1}~,
%\]
%namely, the sample lag-1 autocovariance matrix
%of the time series.
%In a low-dimensional context, $C_{x}$ will be a consistent estimate of
%${\Sigma}_{x}$ so that the number of its non null singular values will
%converge to $m$, thus  provides a consistent estimate.
%In the high-dimensional context, such consistency
%no more holds and the estimation of $m$ becomes a tough
%problem.  No consistent estimate of $m$ has been found so far (see
%\citet{LamYao12} for more details on the issue).
%Notice that (assuming $\mu=0$)
\begin{eqnarray*}
 \Gamma_x &=& \frac1T \sum_{j=1}^T (\Lambda{f}_j+\veps_j)(\Lambda{f}_{j-1}+\veps_{j-1})^t\\
  &=&  \Lambda \left(\frac1T\sum_{j=1}^T {f}_j {f}^t_{j-1}\right)\Lambda^t
  + \Lambda  \left(\frac1T\sum_{j=1}^T {f}_j \veps^t_{j-1}\right)
  + \left(\frac1T\sum_{j=1}^T \veps_j {f}^t_{j-1} \right)\Lambda^t
  + C_{1}~.
\end{eqnarray*}
Because $\Lambda$ has rank $m$,
the first three terms all have rank bounded by $m$ and
$\Gamma_{x}$ appears as a finite-rank perturbation of the lag-1 auto-covariance matrix
$C_1$ which  in general has rank $p\gg m$. Therefore, understanding the properties of the singular
values of $C_1$ will be of primary importance for the
understanding of the $m$ largest singular values of
the matrix of  $\Gamma_{x}$ which are, as said above,
fundamental for the determination of the number of factors
$m$. Notice however that this statistical problem is given here to describe a potential application of the theory established in this paper, but this theory on singular value distribution is general and can be applied to fields other than statistics.

%%%%%%%%%%%%%%%%%%%%%%%%%%%%%%%%%%%%%%%%%%%%%%%%%%%%%%%%%
If we take $\tau=0$ in \eqref{Ceps}, the matrix $S=\frac 1T\sum_{j=1}^T\veps_j\veps_j^t$ is the sample covariance matrix from the observations. The theory for eigenvalue distributions of $S$ has been extensively studied in the random matrix literature dating back to the seminal paper \cite{MP1967}. In this paper, the famous Mar$\check{c}$enko-Pastur law as limit of eigenvalue distributions has been found for a wide class of sample covariance matrices. Further development includes the almost sure convergence of these distributions (\cite{Silv95}) and conditions for convergence of the largest and the smallest eigenvalues, see \cite{BY93}. Meanwhile book-length analysis of sample covariance matrices can be found in \cite{BS10}, \cite{AGZ10}, \cite{PasShc10}. The situation of an auto-covariance matrix $C_\tau$ is completely different. To author's best knowledge, none of the existing literature in random matrix theory treats the sample auto-covariance matrix and the limit for its eigenvalue distribution found in this paper is new.

There are basically two major differences between $C_\tau$ and $S$. First, while $S$ is a non-negative symmetric random matrix, $C_\tau$ is even not symmetric and we must rely on singular value distributions which are in general much more involved than eigenvalue distributions. Secondly, because of the positive lag $\tau$, the summands in $C_\tau$ are no more independent as it is the case for the sample covariance matrix $S$. This again makes the analysis of $C_\tau$ more difficult. As a consequence of these major differences, several new techniques are introduced in the paper in order to complete the proofs, although the general strategy is common in the random matrix theory (see \citet{BS10,PasShc10}). For example, the characterization of the Stieltjes transform of the limiting distribution is obtained via a system of equations due to the time delay $\tau$ where for the case of sample covariance matrix, the characterization is given by a single equation(\cite{MP1967}, \cite{Silv95}).

The rest of the paper is organized as follows. The main theorem of the paper is introduced in Section~\ref{results}.
Section~\ref{proofs} details the proof of the main theorem when time lag $\tau=1$. Section ~\ref{extension} generalizes the proof from time lag $\tau=1$ to any given positive number.
Meanwhile, in contrast to other aspects discussed above, the preliminary steps of truncation, centralization and standardization of the matrix entries are similar to the case of a sample covariance matrix. They are given in Appendix~\ref{app}. To ease the reading of the proof,
technical lemmas are grouped in Section~\ref{lemmas}.

\section{Main Results}
\label{results}

In this paper, we intend to derive the limiting singular value
distribution of the lag$-\tau$ auto-covariance matrix defined in \eqref{Ceps}. It will be done in two steps. We derive the main result first for the lag-1$(\tau=1)$ sample auto-covariance matrix $C_1=\frac 1T\sum_{t=1}^T\veps_j\veps_{j-1}^t$. It turns out that the general case $\tau\geq 1$ is essentially the same and the extension is easily obtained. The details of the extension are given in Section~\ref{extension}.

Therefore, we consider the lag-1 sample auto-covariance matrix $C_1= \frac1T \sum_{j=1}^T \veps_j \veps_{j-1}^t$. By definition, it  is equivalent to study the limiting spectral distribution(LSD) of the matrix
$$A=C_{1}C_{1}^t=\frac{1}{T^2}(\sum_{j=1}^T\varepsilon_j\varepsilon^t_{j-1})(\sum_{j=1}^T\varepsilon_{j-1}\varepsilon^t_j).$$ Alternatively,
\[A=\frac{1}{T^{2}}XY^{t}YX^{t},\] where
$X=\left(\varepsilon_{1},\cdots,\varepsilon_{T}\right)_{p\times
  T}$, $Y=\left(\varepsilon_{0},\cdots,\varepsilon_{T-1}\right)_{p\times
  T}$. Here we define a modified version of the A
matrix,
\[B=\frac{1}{T^{2}}Y^{t}YX^{t}X=\sum_{j=1}^{p}s_{j}s_{j}^{t}\sum_{j=1}^{p}r_{j}r_{j}^{t},
\]
where
$s_{j}=\frac{1}{\sqrt{T}}\left(\varepsilon_{j0},\varepsilon_{j1},\cdots,\varepsilon_{j,T-1}\right)^{t}$
is the j-th row of $Y$, and
$r_{j}=\frac{1}{\sqrt{T}}\left(\varepsilon_{j1},\varepsilon_{j2},\cdots,\varepsilon_{j,T}\right)^{t}$
the j-th row of $X$. As $A$ and $B$ have same nonzero eigenvalues, the
LSD of $A$ can be derived from the LSD of $B$.

The main result of the paper is
\begin{thm}\label{th1}
  Let the following assumptions hold:
  \begin{itemize}
  \item[(a)]
    $\varepsilon_{i}=\left(\varepsilon_{1i},\cdots\varepsilon_{pi}\right)^{t},i=0,1,2,\cdots,T$
    are independent p-dimensional real-valued random vectors with
    independent entries satisfying condition:
    \begin{displaymath}
      \mathbb{E}(\varepsilon_{it})=0,~\mathbb{E}(\varepsilon^2_{it})=1,~\sup_{1\leq i\leq p,0\leq t\leq T}\mathbb{E}\left(|\varepsilon_{it}|^{4}\right)<M,
    \end{displaymath}
   for some constant $M$ and for any $\eta>0$,
    \begin{displaymath}
           \frac{1}{\eta^{4}pT}\sum_{i=1}^{p}\sum_{t=0}^{T}\mathbb{E}\left(|\varepsilon_{it}|^{4}I_{(|\varepsilon_{it}|\geq\eta T^{1/4})}\right)=o\left(1\right);
    \end{displaymath}
  \item[(b)] As $p\rightarrow \infty$, the sample size
    $T=T(p)\rightarrow \infty$ and $p/T\gt c>0$.
  \end{itemize}
  Then,
  \begin{itemize}
  \item[(1)]  as $p,T\rightarrow \infty$, almost surely, the empirical
  spectral distribution $F^B$ of $B$, converges to a non-random
  probability distribution $\b{F}$ whose Stieltjes transform
  $x=x(\alpha)$, $\alpha\in \mathbb{C}\setminus \mathbb{R}$, satisfies
  the equation
  \begin{equation}\label{eq1}
    \alpha^{2}x^{3}-2\alpha\left(c-1\right)x^{2}+\left(c-1\right)^{2}x-\alpha x-1=0.
  \end{equation}
 \item[(2)]  Moreover, for $\alpha\in \mathbb{C}^+=\{z:\mathfrak{Im}z>0\}$, equation \eqref{eq1} admits an unique solution $\alpha\mapsto x(\alpha)$ with positive imaginary part and the density function of the LSD $\b{F}$ is:

  \begin{tabular}{@{}l}
    \begin{minipage}[c]{0.65\textwidth}
      {\begin{align*} f(u)&=\frac{1}{\pi
            u}\left\{-u-\frac{5(c-1)^2}{3}+\frac{2^{4/3}(3u+(c-1)^2)(c-1)}{3d(u)^{1/3}}
            +\frac{2^{2/3}(c-1)d(u)^{1/3}}{3}\right.\\
          &\quad + \left.\frac{1}{48}\left[-8(c-1)+\frac{2\times
                2^{1/3}(3u+(c-1)^2)}{d(u)^{1/3}}+2^{2/3}d(u)^{1/3}\right]^2\right\}^{1/2},
        \end{align*}}
    \end{minipage}
    \\
    where
    $d(u)=-2(c-1)^3+9(1+2c)u+3\sqrt{3}\sqrt{u(-4u^2+(-1+4c(5+2c))u-4c(c-1)^3)}$.
  \end{tabular}

  \vspace{0.25cm} Moreover, the support of f(u) is $(0,b]$ for $0<c\leq 1$, and
  $[a,b]$ for $c>1$, where
  \[
  a=\frac{1}{8}(-1+20c+8c^2-(1+8c)^{3/2}),
  \quad b=\frac{1}{8}(-1+20c+8c^2+(1+8c)^{3/2}).
  \]

  \end{itemize}

\end{thm}

%\begin{proposition}\label{pro1}
%\end{proposition}

It's easy to check that when $c<1$, the
  LSD of $B$ has a point mass $1-c$ at the origin since
  $\mathrm{rank}(B)=p<T$ for large $p$ and $T$, and at the same time
  we have %\renewcommand{\arraystretch}{1.1}
  \begin{equation*}
    \left\{
      \begin{array}{l}\displaystyle
        \int_0^{b}f(u)du=c,\quad 0<c<1,\\[1.5mm]\displaystyle
        \int_{a}^{b}f(u)du=1,\quad c\geq 1.
      \end{array}
    \right.
  \end{equation*}

Since the matrix $A$ we are interested in has the same
non-zero eigenvalues with $B$, the following proposition holds.
\begin{proposition}\label{pro2}
  Under the conditions of Theorem \ref{th1}, the ESD of $A$
  converges a.s. to a non-random limit distribution
  \[F=\frac{1}{c}\b{F}+(1-\frac{1}{c})\delta_0,\] whose Stieltjes
  transform $y=y(\alpha)$, $\alpha\in \mathbb{C}\setminus \mathbb{R}$,
  satisfies the equation
  \[\alpha^2 c^2 y^3+\alpha c(c-1)y^2-\alpha y-1=0.\]
  In particular, $F$ has the density function
  \begin{equation*}
    \left\{
      \begin{array}{l}\displaystyle
        \frac{1}{c}f(u),~u\in (0,b],\mbox{ for }0<c<1,\\[1.5mm]\displaystyle
        \frac{1}{c}f(u),~u\in [a,b],\mbox{ for }c\geq 1.
      \end{array}
    \right.
  \end{equation*}
  where in the later case $c\geq 1$, $F$ has an additional mass
  $(1-\frac{1}{c})$ at the origin.
\end{proposition}

The following details the density function of LSD
of $A$ for different values of c.
\par
% \framebox{\begin{minipage}[t]{1.0\columnwidth}%
% \end{minipage}}
\begin{itemize}
\item When $c=1$, the support is $0\leq u\leq
  \frac{27}{4}$ and the density function is
  \[\frac{1}{c}f(u)=\frac{1}{\pi
    u}\left[-u+3\left(\frac{u}{2^{2/3}d(u)^{1/3}}+\frac{d(u)^{1/3}}{6\times2^{1/3}}\right)^2\right]^{1/2},\]
  where $d(u)=27u+3\sqrt{3}\times\sqrt{u(-4u^2+27u)}$. It's easy to
  see that as $u\rightarrow 0_{+}$, $f(u)\rightarrow\infty$.

\item If $c<1$, it can be seen from the explicit form of $f(u)$ that
  when $u\gt 0_{+}$, $\frac{1}{c}f(u)\rightarrow\infty$ because the $u$ in the
  denominator of the density function cannot be completely canceled
  out.

\item If $c>1$, the shape of the density function turns out to be a
  little different from the case $c\leq 1$. Nevertheless it's quite
  intuitive because the lower bound of the support is positive and the
  density function is bounded.\end{itemize}

   The density functions of LSD of A for $c=0.5,1,2,3$
  are displayed on Figure \ref{fig1}.

  \begin{figure}[ht]
    % Requires \usepackage{graphicx}
    \centering
    \includegraphics[width=1.0\textwidth]{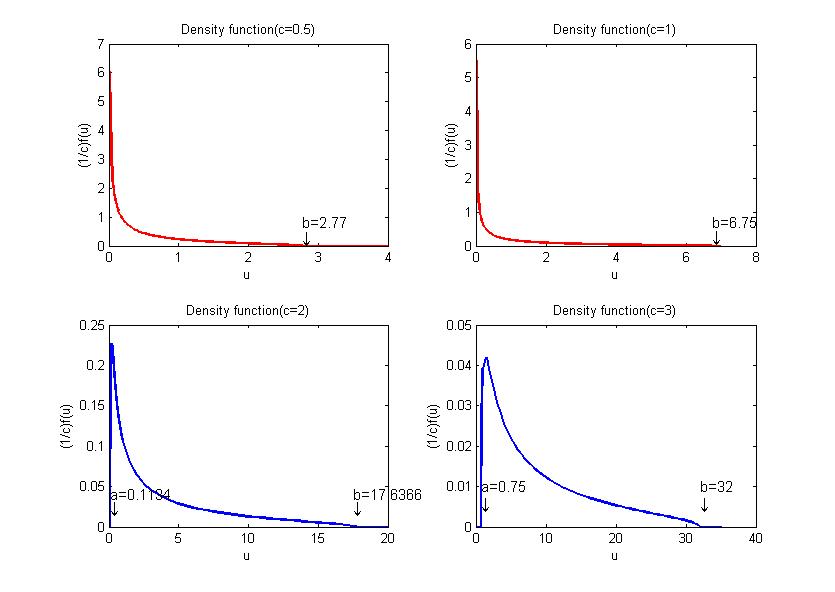}\\
    \caption{Density plots of the LSD of B.Top to bottom and left to
      right: c=0.5,1,2 and 3, respectively}
    \label{fig1}
  \end{figure}

\section{Proofs}
\label{proofs}

\subsection{Proof of Theorem \ref{th1}}
The proof of the theorem follows the general strategy based on the
Stieltjes transform as presented
in \citet{Silv95}, \citet{BS10} and
\citet{PasShc10}.
However, the random matrix B here is
no more a covariance matrix as considered in these references. Almost
all the steps of the proof need new arguments and ideas compared to
the case of sample covariance matrices considered so far in the literature. Following
this method, the first step is to truncate the entries
$\{\varepsilon_{jt}\}$ at a convenient rate using Assumption
(a). After truncation and the follow-up steps of centralization and
standardization, we may assume that
\[
|\varepsilon_{ij}|\leq\eta T^{1/4},\quad
\mathbb{E}\left(\varepsilon_{ij}\right)=0,\quad
Var\left(\varepsilon_{ij}\right)=1,\quad \sup_{1\leq i\leq p,0\leq
  j\leq T}\mathbb{E}\left(|\varepsilon_{ij}|^{4}\right)<M.
\]
\noindent
The details of these technical steps are given in Appendix A.

\indent By the rank inequality(Theorem A.44 of \cite{BS10}), it is enough to consider
$$B=\sum_{j=1}^{p}s_{j}s_{j}^{t}\sum_{j=1}^{p}r_{j}r_{j}^{t}=P_1\tilde{C}P_1^t\tilde{C},$$
where
\[
s_{j}=P_{1}r_{j}=\frac{1}{\sqrt{T}}(0,\varepsilon_{j1},\cdots,\varepsilon_{j,T-1})^t,
\quad C=\sum_{j=1}^ps_js_j^t, \quad \tilde{C}=\sum_{j=1}^{p}r_{j}r_{j}^{t},\quad
P_{1}=\left(\begin{array}{cc}
    \textbf{0} & 0\\
    \textbf{I}_{T-1} & \textbf{0}
  \end{array}\right).
\]
\noindent
At this stage, the important observation is that here we have replaced
$s_j=\frac{1}{\sqrt{T}}(\varepsilon_{j0},\varepsilon_{j1},\cdots,\varepsilon_{j,T-1})^t$
by
$\tilde{s}_j=\frac{1}{\sqrt{T}}(0,\varepsilon_{j1},\cdots,\varepsilon_{j,T-1})^t$
without altering the LSD of B since when $T\gt\infty$, the effect of
this substitution will vanish. For the sake of convenience, we still
use $s_j$ to denote $\tilde{s}_j$.

For $\alpha\in \mathbb{C}\setminus\mathbb{R}$, define
$$B\left(\alpha\right)=\sum_{j=1}^{p}s_{j}s_{j}^{t}\sum_{j=1}^{p}r_{j}r_{j}^{t}-\alpha I_{T}.$$
Let
$$x_0=\frac{1}{T}tr(B^{-1}(\alpha)),\quad y_0=\frac{1}{T}tr(\tilde{C}B^{-1}(\alpha)),\quad z_0=\frac{1}{T}tr(B^{-1}(\alpha)C).$$
The method consists in finding a system of two asymptotic equations
satisfied by $x_0$ and $y_0$. Solving the system yields an asymptotic
equivalent for $x_0$ and finally leads to the equation \eqref{eq1} satisfied
by the limit of $x_0$. Nonetheless, $x_0$ is the Stieltjes transform of the matrix B which can be recovered from the inversion formula.

Let
$$
B_{j}\left(\alpha\right)=\sum_{k\neq j}s_{k}s_{k}^{t}\sum_{i\neq
  j}r_{i}r_{i}^{t}-\alpha I_{T},\quad C_j=C-s_js_j^t,\quad \tilde{C}_j=\tilde{C}-r_jr_j^t,\quad 1\leq j\leq
p,$$ then
\begin{align*}
B\left(\alpha\right)&=B_{j}\left(\alpha\right)+\sum_{i\neq
  j}s_{j}s_{j}^{t}r_{i}r_{i}^{t}+\sum_{k\neq
  j}s_{k}s_{k}^{t}r_{j}r_{j}^{t}+s_{j}s_{j}^{t}r_{j}r_{j}^{t}\\
  &=B_j\left(\alpha\right)+s_js_j^t\tilde{C}_j+C_jr_jr_j^t+s_js_j^tr_jr_j^t.
\end{align*}

We have
\[
I_{T}=B(\alpha)B^{-1}(\alpha)=\left(\sum_{j=1}^{p}s_{j}s_{j}^{t}\right)\left(\sum_{j=1}^{p}r_{j}r_{j}^{t}\right)B^{-1}\left(\alpha\right)-\alpha
B^{-1}\left(\alpha\right).
\]
\noindent
Taking trace and dividing both sides by $T$, we get

\begin{equation}\label{eq:1}
  1=\frac{1}{T}\sum_{j=1}^{p}s_{j}^{t}\tilde{C}B^{-1}\left(\alpha\right)s_{j}-\alpha\frac{1}{T}tr\left(B^{-1}\left(\alpha\right)\right).
\end{equation}
Note that $x_0=\frac{1}{T}tr(B^{-1}(\alpha))$ is the Stieltjes
transform of the ESD of the matrix B, and its limit will be found by
letting $p,T\rightarrow\infty$ on both sides of the equation.

Consider  $s_{j}^{t}\tilde{C} B^{-1}\left(\alpha\right)s_{j}$, using the identities
\[
\left(B+\sum_{j=1}^{m}ab_{j}^{t}\right)^{-1}a=\frac{B^{-1}a}{1+\sum_{j=1}^{m}b_{j}^{t}B^{-1}a},
\]
and
\[
B^{-1}-D^{-1}=B^{-1}\left(D-B\right)D^{-1},
\]
\noindent
we have
\begin{align*}
  s_{j}^{t}\tilde{C}B^{-1}\left(\alpha\right)s_{j}&=\frac{s_{j}^{t}\tilde{C}\left(B_{j}\left(\alpha\right)+C_{j}r_jr_j^t\right)^{-1}s_{j}}{1+s_{j}^{t}\tilde{C}\left(B_{j}\left(\alpha\right)+C_{j}r_jr_j^t\right)^{-1}s_{j}}\\
 &=1-\frac{1}{1+s_j^t\tilde{C}_j\left(B_{j}\left(\alpha\right)+C_{j}r_jr_j^t\right)^{-1}s_j+s_j^tr_jr_j^t\left(B_{j}\left(\alpha\right)+C_{j}r_jr_j^t\right)^{-1}s_j}\\
  & := 1- \frac{1}{1+L_1+L_2},
\end{align*}
\noindent
where $L_1$ and $L_2$ are explicitly defined.

\noindent
For $L_1$, by Lemma \ref{lem3}, or equivalently by Lemma 2.7 of \cite{BS98}, we have
\begin{eqnarray*}
  L_1&=&s_{j}^{t}\tilde{C}_j\left(B_{j}\left(\alpha\right)+C_{j}r_jr_j^t\right)^{-1}s_{j} \\
  & = & s_{j}^{t}\tilde{C}_j B_{j}^{-1}\left(\alpha\right)s_{j}- s_{j}^{t}\tilde{C}_j B_{j}\left(\alpha\right)^{-1}C_{j}r_jr_j^t\left(B_{j}\left(\alpha\right)+C_{j}r_jr_j^t\right)^{-1}s_{j}\\
  & = & s_{j}^{t}\tilde{C}_j B_{j}^{-1}\left(\alpha\right)s_{j}-\frac{s_{j}^{t}\tilde{C}_j B_{j}^{-1}\left(\alpha\right)C_{j}r_jr_j^tB_{j}\left(\alpha\right)^{-1}s_{j}}{1+r_j^tB_{j}^{-1}\left(\alpha\right)C_{j}r_j}\\
  & = & \frac{1}{T}tr\left(\tilde{C}_j B_{j}^{-1}\left(\alpha\right)\right)-\frac{\frac{1}{T}tr\left(\tilde{C}_j B_{j}^{-1}\left(\alpha\right)C_{j}P_{1}^{t}\right)\cdot\frac{1}{T}tr\left(B_{j}^{-1}\left(\alpha\right)P_{1}\right)}{1+\frac{1}{T}tr\left(B_{j}\left(\alpha\right)^{-1}C_{j}\right)}+o_{a.s.}(1).
\end{eqnarray*}
\noindent

For $L_2$,  we have
\begin{align*}
  L_2&=s_{j}^{t}r_{j}r_{j}^{t}\left(B_{j}\left(\alpha\right)+C_{j}r_jr_j^t\right)^{-1}s_{j}=s_{j}^{t}r_{j}r_{j}^{t}B_{j}^{-1}\left(\alpha\right)s_{j}-\frac{s_{j}^{t}r_{j}r_{j}^{t}B_{j}^{-1}\left(\alpha\right)C_jr_jr_j^tB^{-1}_j(\alpha)s_{j}}{1+r_j^tB_{j}^{-1}\left(\alpha\right)C_{j}r_{j}}\\
  &=\left(s_{j}^{t}P_{1}^{t}s_{j}\right)\cdot\frac{1}{T}tr\left(B_{j}^{-1}\left(\alpha\right)P_{1}\right)-\frac{\left(s_{j}^{t}P_{1}^{t}s_{j}\right)\cdot\frac{1}{T}tr\left(B_{j}^{-1}(\alpha)C_j\right)\cdot\frac{1}{T}tr\left(B_{j}^{-1}\left(\alpha\right)P_{1}\right)}{1+\frac{1}{T}tr\left(B_{j}^{-1}\left(\alpha\right)C_{j}\right)}+o_{a.s.}(1)=o_{a.s.}(1).
\end{align*}
\noindent

Therefore, by equation $\eqref{eq:1}$, we have
\begin{align}
&~1+\alpha\frac{1}{T}tr(B^{-1}(\alpha))=o_{a.s.}(1)+ \label{yt1}\\
&\frac{p}{T}\left(1-\dfrac{1+\dfrac{1}{T}tr(B^{-1}(\alpha)C)}{\left(1+\dfrac{1}{T}tr\left(B^{-1}(\alpha)C\right)\right)\left(1+\dfrac{1}{T}tr(\tilde{C}B^{-1}(\alpha))\right)-\dfrac{1}{T}tr\left(\tilde{C} B^{-1}\left(\alpha\right)CP_{1}^{t}\right)\cdot\dfrac{1}{T}tr\left(B^{-1}\left(\alpha\right)P_{1}\right)}\right)\nonumber
\end{align}

Here, we have used the following equivalents, uniformly in $j$, as $p,T\rightarrow\infty$,
$$\frac{1}{T}tr\left(B_{j}^{-1}\left(\alpha\right)C_{j}\right)=z_{0}+o_{a.s.}(1),$$
$$\frac{1}{T}tr\left(B_{j}^{-1}\left(\alpha\right)\right)=x_0+o_{a.s.}(1),$$ $$\frac{1}{T}tr\left(\tilde{C}_jB_{j}^{-1}\left(\alpha\right)\right)=y_{0}+o_{a.s.}(1).$$
\noindent

Similar to equation \eqref{eq:1}, we have

\begin{equation}
  1=\frac{1}{T}\sum_{j=1}^{p}r_{j}^{t}B^{-1}\left(\alpha\right)Cr_{j}-\alpha\frac{1}{T}tr\left(B^{-1}\left(\alpha\right)\right).\label{eq:2}
\end{equation}

Considering  $r_{j}^{t} B^{-1}\left(\alpha\right)C r_{j}$, we have
\begin{align*}
  r_{j}^{t}B^{-1}\left(\alpha\right)Cr_{j}&=\frac{r_{j}^{t}\left(B_{j}\left(\alpha\right)+s_j s_j^t \tilde{C}_{j}\right)^{-1}Cr_{j}}{1+r_{j}^{t}\left(B_{j}\left(\alpha\right)+s_j s_j^t \tilde{C}_{j}\right)^{-1}Cr_{j}}\\
 &=1-\frac{1}{1+r_j^t\left(B_{j}\left(\alpha\right)+s_j s_j^t\tilde{C}_j\right)^{-1}C_jr_j+r_j^t\left(B_{j}\left(\alpha\right)+s_js_j^t\tilde{C}_j\right)^{-1}s_j s_j^tr_j}\\
  & := 1- \frac{1}{1+W_1+W_2},
\end{align*}
\noindent
where $W_1$ and $W_2$ are explicitly defined.

\noindent
For $W_1$, we have
\begin{eqnarray*}
  W_1&=&r_{j}^{t}\left(B_{j}\left(\alpha\right)+s_j s_j^t\tilde{C}_j\right)^{-1}C_{j}r_{j} \\
  & = & r_{j}^{t}B_{j}^{-1}\left(\alpha\right)C_jr_{j}- r_{j}^{t} B_{j}^{-1}\left(\alpha\right)s_j s_j^t\tilde{C}_j\left(B_{j}\left(\alpha\right)+s_j s_j^t\tilde{C}_j\right)^{-1}C_{j}r_{j}\\
  & = & r_{j}^{t}B_{j}^{-1}\left(\alpha\right)C_jr_{j}-\dfrac{r_j^tB_{j}^{-1}\left(\alpha\right)s_{j}s_{j}^{t}\tilde{C}_j B_{j}^{-1}\left(\alpha\right)C_{j}r_j}{1+s_j^t \tilde{C}_{j}B_{j}^{-1}\left(\alpha\right)s_j}\\
  & = & \frac{1}{T}tr\left(B_{j}^{-1}\left(\alpha\right)C_j\right)-\frac{\frac{1}{T}tr\left(\tilde{C}_j B_{j}^{-1}\left(\alpha\right)C_{j}P_{1}^{t}\right)\cdot\frac{1}{T}tr\left(B_{j}^{-1}\left(\alpha\right)P_{1}\right)}{1+\frac{1}{T}tr\left(\tilde{C}_{j}B_{j}\left(\alpha\right)^{-1}\right)}+o_{a.s.}(1).
\end{eqnarray*}

\noindent
For $W_2$,  we have
\begin{align*}
  W_2&=r_{j}^{t}\left(B_{j}\left(\alpha\right)+s_js_j^t\tilde{C}_{j}\right)^{-1}s_{j}s_{j}^{t}r_{j}=r_{j}^{t}B_{j}^{-1}\left(\alpha\right)s_{j}s_{j}^{t}r_{j}-\frac{r_{j}^{t}B_{j}^{-1}\left(\alpha\right)s_js_j^t\tilde{C}_jB^{-1}_j(\alpha)s_{j}s_{j}^{t}r_{j}}{1+s_j^t\tilde{C}_jB^{-1}_j(\alpha)s_{j}}\\
  &=\left(s_{j}^{t}P_{1}^{t}s_{j}\right)\cdot\frac{1}{T}tr\left(B_{j}^{-1}\left(\alpha\right)P_{1}\right)-\frac{\left(s_{j}^{t}P_{1}^{t}s_{j}\right)\cdot\frac{1}{T}tr\left(\tilde{C}_j B_{j}^{-1}(\alpha)\right)\cdot\frac{1}{T}tr\left(B_{j}^{-1}\left(\alpha\right)P_{1}\right)}{1+\frac{1}{T}tr\left(\tilde{C}_{j}B_{j}^{-1}\left(\alpha\right)\right)}+o_{a.s.}(1)=o_{a.s.}(1).
\end{align*}
\noindent

Therefore, by equation $\eqref{eq:2}$, we have
\begin{align}
&~1+\alpha\frac{1}{T}tr(B^{-1}(\alpha))=o_{a.s.}(1)+ \label{yt2}\\
&\frac{p}{T}\left(1-\dfrac{1+\dfrac{1}{T}tr(B^{-1}(\alpha)\tilde{C})}{\left(1+\dfrac{1}{T}tr\left(B^{-1}(\alpha)C\right)\right)\left(1+\dfrac{1}{T}tr(\tilde{C}B^{-1}(\alpha))\right)-\dfrac{1}{T}tr\left(\tilde{C} B^{-1}\left(\alpha\right)CP_{1}^{t}\right)\cdot\dfrac{1}{T}tr\left(B^{-1}\left(\alpha\right)P_{1}\right)}\right)\nonumber
\end{align}

\noindent
Thus, according to equation \eqref{yt1} and \eqref{yt2}, we have
\[\dfrac{1}{T}tr(B^{-1}(\alpha)\tilde{C})=\dfrac{1}{T}tr(B^{-1}(\alpha)C)+o_{a.s.}(1).\]

By Lemma \ref{lem1}, the second term is $o_{a.s.}(1)$ since both
$\frac{1}{T}tr\left(P_{1}^{t}\tilde{C}_j B_{j}\left(\alpha\right)^{-1}C_{j}\right)$
and $\frac{1}{T}tr\left(B_{j}\left(\alpha\right)^{-1}C_{j}\right)$ are
non-negative and bounded as $p,T\rightarrow\infty$.
\[
  L_1 =\frac{1}{T} tr\left(\tilde{C}_jB_j^{-1}(\alpha)\right)+o_{a.s.}(1)=y_0+o_{a.s.}(1).
\]
\noindent

Finally, by equation $\eqref{eq:2}$, we find
\begin{equation}\label{eq3}
  1+\alpha x_0=\frac{p}{T}\left(1-\frac{1}{1+y_0}\right)+o_{a.s.}(1).
\end{equation}

To find a second equation satisfied by $x_0$ and $y_0$, using Lemma
\ref{lem3} and Lemma \ref{lem1},

\begin{align*}
  \frac{1}{T}tr(\tilde{C}B^{-1}(\alpha))&=\frac{1}{T}tr(\sum_{j=1}^pr_jr_j^tB^{-1}(\alpha))=\frac{1}{T}\sum_{j=1}^pr_j^tB^{-1}(\alpha)r_j\\
  &=\frac{1}{T}\sum_{j=1}^p\dfrac{r_j^t\left(B_j(\alpha)+s_js_j^t\tilde{C}_j\right)^{-1}r_j}{1+r_j^t\left(B_j(\alpha)+s_js_j^t\tilde{C}_j\right)^{-1}C_jr_j+r_j^t\left(B_j(\alpha)+s_js_j^t\tilde{C}_j\right)^{-1}s_js_j^tr_j}\\
   &=\frac{1}{T}\sum_{j=1}^p\dfrac{r_j^tB_j^{-1}(\alpha)r_j-\dfrac{r_j^tB_j^{-1}(\alpha)s_js_j^t\tilde{C}_jB_j^{-1}(\alpha)r_j}{1+s_j^t\tilde{C}_jB_j^{-1}(\alpha)s_j}}{1+r_j^tB^{-1}_j(\alpha)C_jr_j-\dfrac{r_j^tB^{-1}_j(\alpha)s_js_j^t\tilde{C}_jB^{-1}_j(\alpha)C_jr_j}{1+s_j^t\tilde{C}_jB^{-1}_j(\alpha)s_j}}+o_{a.s.}(1)\\
  &=\frac{p}{T}\cdot\dfrac{\dfrac{1}{T}tr(B^{-1}(\alpha))}{1+\dfrac{1}{T}tr(B^{-1}(\alpha)C)}+o_{a.s.}(1).
\end{align*}
\noindent
This leads to
\begin{equation}\label{eq4}
  y_0=\frac{p}{T}\cdot\frac{x_0}{1+y_0}+o_{a.s.}(1).
\end{equation}

In conclusion, $(x_0,y_0)$ satisfy the system

\[
\begin{cases}\displaystyle
  1+\alpha x_0=\frac{cy_0}{1+y_0}+o_{a.s.}(1),\\[1.5mm]\displaystyle
  y_0=\frac{cx_0}{1+y_0}+o_{a.s.}(1).
\end{cases}
\]
\noindent
Notice that for any $T,~|x_0|\leq \frac{1}{|\mathfrak{Im}(\alpha)|}$ is bounded, and by equation \eqref{eq4}, $|y_0|$ is also bounded as $T\rightarrow \infty$, otherwise \eqref{eq4} may not hold. Therefore, both $\{x_0\}$ and $\{y_0\}$ are bounded sequences. Let be two subsequences $\{x_{t_n}\},\{y_{t_n}\}$ so that $x_{t_n}\rightarrow x$ and $y_{t_n}\rightarrow y$ as $n\rightarrow\infty$. It can be concluded that the limiting functions
$(x,y)$ satisfy the system of equations:

\[
\begin{cases}\displaystyle
  1+\alpha x=\frac{cy}{1+y}&(1)\\[1.5mm]\displaystyle
  y=\frac{cx}{1+y}&(2)
\end{cases}
\]
\noindent

By eliminating $y$, we finally find the equation \eqref{eq1} satisfied by the limiting
function $x$. Denote by $\mathcal{F}$ all the analytical functions $\{f:~\mathbb{C}^+\mapsto \mathbb{C}^+\}$. Because according to the following proof we have one unique solution on $\mathcal{F}$ that satisfies equation \eqref{eq1}, the whole bounded sequence $\{x_0\}$ has one unique limit $x$ in $\mathcal{F}$.

%\begin{equation}
%  \alpha^{2}x^{3}-2\alpha\left(c-1\right)x^{2}+\left(c-1\right)^{2}x-\alpha x-1=0.\label{eq:3}
%\end{equation}
%\noindent
%%To conclude the proof, we need to check that only one of the solutions
%%of equation \eqref{eq1} corresponds to the Stieltjes Transform of a
%%finite measure and this measure is non-defective, i.e. a probability
%%measure.

%\subsection{Proof of Proposition \ref{pro1}}

As for the second statement of Theorem \ref{th1}, in order to find the density function of the LSD $\b{F}$ of
$B$, we use the inversion formula:

\[
f\left(u\right)=\lim_{\varepsilon\gt 0_{+}}\frac{1}{\pi}{\mathfrak{Im}
  }x\left(u+i\varepsilon\right)
\]
\noindent
where $x\left(\cdot\right)$ is the Stieltjes transform of $\b{F}$. Write
$$\lim_{\varepsilon\gt 0_{+}}x(u+i\varepsilon)=\psi(u)+i\phi(u),$$
both $\psi$ and $\phi$ are real-valued functions of $u$. By
substituting $\alpha=u+i\varepsilon$, $x=\psi+i\phi$ into equation
\eqref{eq1} and letting $\varepsilon\gt 0_{+}$, both the real part and the imaginary part of the LHS of equation \eqref{eq1}
should equal to 0, i.e.
\[
\begin{cases}\displaystyle
  u^{2}{\psi}^{3}-3u^{2}\psi\cdot {\phi}^{2}-2u\left(c-1\right)\left({\psi}^{2}-{\phi}^{2}\right)-\left(u-\left(c-1\right)^{2}\right)\psi-1=0 & \left(3\right)\\[1.5mm]\displaystyle
  -u^{2}{\phi}^{2}+3u^{2}{\psi}^{2}-4u\left(c-1\right)\psi-\left(u-\left(c-1\right)^{2}\right)=0
  & \left(4\right)
\end{cases}
\]
By plugging in (4) into (3), we get
\[-8u^2\psi^3+16u(c-1)\psi^2+(2u-10(c-1)^2)\psi+\frac{2(c-1)^3}{u}-2c+1=0.\]
Solving this equation and substituting for $\psi$ in (4), we get
\begin{align*}
  {\phi}_1^2(u)&=\frac{1}{u^2}\left\{-u-\frac{5(c-1)^2}{3}+\frac{2^{4/3}(3u+(c-1)^2)(c-1)}{3d(u)^{1/3}}+\frac{2^{2/3}(c-1)d(u)^{1/3}}{3}\right.\\
  &\quad +\left.\frac{1}{48}\left[-8(c-1)+\frac{2\times
        2^{1/3}(3u+(c-1)^2)}{d(u)^{1/3}}+2^{2/3}d(u)^{1/3}\right]^2\right\},
\end{align*}
\begin{align*}
  {\phi}_2^2(u)&=\frac{1}{u^2}\left\{-u-\frac{5(c-1)^2}{3}+\frac{1+i\sqrt{3}}{2}\cdot\frac{2^{4/3}(3u+(c-1)^2)(c-1)}{3d(u)^{1/3}}+\frac{1-i\sqrt{3}}{2}\cdot\frac{2^{2/3}(c-1)d(u)^{1/3}}{3}\right.\\
  &\quad +\left.\frac{1}{48}\left[-8(c-1)+\frac{1+i\sqrt{3}}{2}\cdot\frac{2\times
        2^{1/3}(3u+(c-1)^2)}{d(u)^{1/3}}+\frac{1-i\sqrt{3}}{2}\cdot2^{2/3}d(u)^{1/3}\right]^2\right\},
\end{align*}
\begin{align*}
  {\phi}_3^2(u)&=\frac{1}{u^2}\left\{-u-\frac{5(c-1)^2}{3}+\frac{1-i\sqrt{3}}{2}\cdot\frac{2^{4/3}(3u+(c-1)^2)(c-1)}{3d(u)^{1/3}}+\frac{1+i\sqrt{3}}{2}\cdot\frac{2^{2/3}(c-1)d(u)^{1/3}}{3}\right.\\
  &\quad +\left.\frac{1}{48}\left[-8(c-1)+\frac{1-i\sqrt{3}}{2}\cdot\frac{2\times
        2^{1/3}(3u+(c-1)^2)}{d(u)^{1/3}}+\frac{1+i\sqrt{3}}{2}\cdot2^{2/3}d(u)^{1/3}\right]^2\right\},
\end{align*}
where
\begin{equation}\label{eq5}
  d(u)=-2(c-1)^3+9(1+2c)u+3\sqrt{3}\sqrt{u(-4u^2+(-1+4c(5+2c))u-4c(c-1)^3)}.
\end{equation}
It can be checked that only the first solution is
compatible with the fact that both $\psi$ and $\phi$ are real-valued
functions of $u$, i.e.
\begin{align*}
  {\phi}^2(u)&=\frac{1}{u^2}\left\{-u-\frac{5(c-1)^2}{3}+\frac{2^{4/3}(3u+(c-1)^2)(c-1)}{3d(u)^{1/3}}+\frac{2^{2/3}(c-1)d(u)^{1/3}}{3}\right.\\
  &\quad +\left.\frac{1}{48}\left[-8(c-1)+\frac{2\times
        2^{1/3}(3u+(c-1)^2)}{d(u)^{1/3}}+2^{2/3}d(u)^{1/3}\right]^2\right\}.
\end{align*}

From the explicit form of ${\phi}^2(u)$ we see that, necessarily,
\[u(-4u^2+(-1+4c(5+2c))u-4c(c-1)^3)\geq 0,\] since $u\geq 0$. Solving
this quadratic inequality, we get two roots,
\begin{equation}\label{eq6}
  a=\frac{1}{8}(-1+20c+8c^2-(1+8c)^{3/2}),\quad b=\frac{1}{8}(-1+20c+8c^2+(1+8c)^{3/2}).
\end{equation}
It's very easy to see that $a$ is an increasing function of $c$
and $a=0$ when $c=1$.

In other words, if $0<c<1$, $-\frac{1}{4}<a<0$, then the support of
the density function should be $(0,b)$. If $c\geq 1$, $a\geq 0$, then
the support of the density function is $(a,b)$.

Then the density function of the limiting spectral distribution of the
$T\times T$ dimensional multiplied lag-1 sample auto-covariance
matrix $B$ is

\begin{tabular}{@{}l}
  \begin{minipage}[c]{0.65\textwidth}
    {\begin{align*} f(u)&=\frac{1}{\pi
          u}\left\{-u-\frac{5(c-1)^2}{3}+\frac{2^{4/3}(3u+(c-1)^2)(c-1)}{3d(u)^{1/3}}
          +\frac{2^{2/3}(c-1)d(u)^{1/3}}{3}\right.\\
        &\quad + \left.\frac{1}{48}\left[-8(c-1)+\frac{2\times
              2^{1/3}(3u+(c-1)^2)}{d(u)^{1/3}}+2^{2/3}d(u)^{1/3}\right]^2\right\}^{1/2},
      \end{align*}}
  \end{minipage}
\end{tabular}

\noindent
where $0<u\leq b$, for $0<c\leq 1$ and $a \leq u \leq b$, for $c>1$,
with $(a,b)$ given in equation \eqref{eq5} and $d(u)$ given in equation
\eqref{eq6}. Therefore, equation \eqref{eq1} admits at least one solution $\alpha\mapsto x(\alpha)$ that corresponds to this density function of the LSD $\b{F}$. As for the uniqueness, suppose there exists another solution $x_1(\alpha)$ that satisfies equation \eqref{eq1}, then there should be another density $f_1(u)$ that corresponds to $x_1(\alpha)$ while $f_1(u)\neq f(u)$. However, it can be seen from the previous deductions that the density function is unique. Therefore, $f_1(u)=f(u)$, $x_1(\alpha)=x(\alpha)$. Equation \eqref{eq1} admits one unique solution.

\subsection{Proof of Proposition \ref{pro2}}

Under the same conditions in \textbf{Theorem \ref{th1}}, the ESD of
$A$ converges to a non-random limit distribution $F$ with Stieltjes
transform $y=y(\alpha)$. On the other hand, the ESD of $B$ converges
to $\b{F}$ with Stieltjes transform $x=x(\alpha)$ satisfying
\[\alpha^2x^3-2\alpha(c-1)x^2+(c-1)^2x-\alpha x-1=0.\]

Since it's known that
\[F=\frac{1}{c}\b{F}+(1-\frac{1}{c})\delta_0,\] conclusively we have
\[(1-c)(-\frac{1}{\alpha})+cy(\alpha)=x(\alpha).\] Substituting into
the equation of $x$ we can get the equation of $y$, which is
\[\alpha^2 c^2 y^3+\alpha c(c-1)y^2-\alpha y-1=0.\]

\section{Extension to lag-$\tau$ sample auto-covariance matrix}
\label{extension}
So far in previous sections, we have focused on the singular value distribution of the lag-1 sample auto-covariance matrix $C_1=T^{-1}\sum_{j=1}^T\veps_j\veps_{j-1}^t$, while in this section, for any given positive integer $\tau$, we discuss the singular value distribution of the lag-$\tau$ sample auto-covariance matrix $C_{\tau}=T^{-1}\sum_{j=1}^T\veps_j\veps_{j-\tau}^t$.

Here we follow exactly the same strategy used in the derivation of the LSD of the lag-1 sample auto-covariance matrix. It's easy to see that the difference between $C_1$ and $C_\tau$ lies in that we have now for $C_\tau$,
\[s_j=P_1^{\tau}r_j=\frac{1}{\sqrt{T}}(\underbrace{0,\cdots,0,}_{\tau~ 0's}\veps_{j1},\cdots,\veps_{j,T-\tau}),~B=\sum_{j=1}^ps_js_j^t\sum_{j=1}^pr_jr_j^t=P_1^{\tau}\tilde{C}(P_1^{\tau})^t\tilde{C}.\]
\noindent
Meanwhile, the other matrices and notations remain the same using however the new definition of the $s_j's$ above. Consequently, equation \eqref{yt1} becomes
\begin{align}
&~1+\alpha\frac{1}{T}tr(B^{-1}(\alpha))=o_{a.s.}(1)+ \label{yt3}\\
&\frac{p}{T}\left(1-\dfrac{1+\dfrac{1}{T}tr(B^{-1}(\alpha)C)}{\left(1+\dfrac{1}{T}tr\left(B^{-1}(\alpha)C\right)\right)\left(1+\dfrac{1}{T}tr(\tilde{C}B^{-1}(\alpha))\right)-\dfrac{1}{T}tr\left(\tilde{C} B^{-1}\left(\alpha\right)C\left(P_{1}^{\tau}\right)^{t}\right)\cdot\dfrac{1}{T}tr\left(B^{-1}\left(\alpha\right)P_{1}^{\tau}\right)}\right)\nonumber
\end{align}
\noindent
Equation \eqref{yt2} becomes

\begin{align}
&~1+\alpha\frac{1}{T}tr(B^{-1}(\alpha))=o_{a.s.}(1)+ \label{yt4}\\
&\frac{p}{T}\left(1-\dfrac{1+\dfrac{1}{T}tr(B^{-1}(\alpha)\tilde{C})}{\left(1+\dfrac{1}{T}tr\left(B^{-1}(\alpha)C\right)\right)\left(1+\dfrac{1}{T}tr(\tilde{C}B^{-1}(\alpha))\right)-\dfrac{1}{T}tr\left(\tilde{C} B^{-1}\left(\alpha\right)C\left(P_{1}^{\tau}\right)^{t}\right)\cdot\dfrac{1}{T}tr\left(B^{-1}\left(\alpha\right)P_{1}^{\tau}\right)}\right)\nonumber
\end{align}

\noindent
Thus, according to equation \eqref{yt3} and \eqref{yt4}, we still have
\[\dfrac{1}{T}tr(B^{-1}(\alpha)\tilde{C})=\dfrac{1}{T}tr(B^{-1}(\alpha)C)+o_{a.s.}(1).\]
Meanwhile, by Lemma \ref{lem2}, we still have
\begin{equation}\label{lemeq}
  \frac{1}{T}tr\left(B^{-1}(\alpha)P_1^{\tau}\right)=o_{a.s.}(1),
\end{equation}
then by equation \eqref{yt3}, we have
\begin{equation}\label{teq3}
  1+\alpha x_0=\frac{p}{T}\left(1-\frac{1}{1+y_0}\right)+o_{a.s.}(1).
\end{equation}
Similarly, as for the second equation satisfied by $x_0$ and $y_0$, equation \eqref{eq4} persists.
\begin{equation}\label{teq4}
  y_0=\frac{p}{T}\cdot\frac{x_0}{1+y_0}+o_{a.s.}(1).
\end{equation}
 Therefore, the system of equations satisfied by $x_0$ and $y_0$ remains the same when the time lag changes from 1 to $\tau$. In other words, for a given positive time lag $\tau$, the singular value distribution of $C_{\tau}$ is the same with that of $C_1$ established in Theorem \ref{th1}.

\section{TECHNICAL LEMMAS}
\label{lemmas}

\begin{lemma}\label{lem3}
  Under the same assumptions in \textbf{Theorem \ref{th1}}, we have,
  $\forall 1\leq j\leq p$, almost surely,
  \begin{equation}\label{eq7}
    s_j^tB_j^{-1}(\alpha)s_j=\frac{1}{T}tr(B_j^{-1}(\alpha))+o_{a.s.}(1),
  \end{equation}
  \begin{equation}\label{eq8}
    r_j^tB_j^{-1}(\alpha)P_1^kr_j=\frac{1}{T}tr(B_j^{-1}(\alpha)P_1^k)+o_{a.s.}(1),
  \end{equation}
    \begin{equation}\label{eq10}
    r_j^t\tilde{C}_jB_j^{-1}(\alpha)P_1^kr_j=\frac{1}{T}tr(\tilde{C}_jB_j^{-1}(\alpha)P_1^k)+o_{a.s.}(1),
  \end{equation}
  \begin{equation}\label{eq9}
    s_j^tB_j^{-1}(\alpha)C_js_j=\frac{1}{T}tr(B_j^{-1}(\alpha)C_j)+o_{a.s.}(1),
  \end{equation}
  where the $o_{a.s.}(1)$ terms are uniform in $1\leq j\leq p$.
\end{lemma}

\begin{proof}

  We detail the proof of \eqref{eq7} and the proofs of \eqref{eq8}, \eqref{eq10} and
  \eqref{eq9} are very similar, thus omitted.

  Denote $B_j^{-1}(\alpha)$ by $(y_{kl})=Y$,
  $s_j=\frac{1}{\sqrt{T}}(\varepsilon_{j0},\cdots,\varepsilon_{j,T-1})^t$,
  then we have
  \[|y_{kl}|<\frac{1}{\nu},\quad |\varepsilon_{it}|<\eta
  T^{\frac{1}{4}},\quad \sup_{1\leq i\leq p,0\leq t\leq
    T}\mathbb{E}|\varepsilon_{it}|^4<M,\] where $\nu$ is the image
  part of $\alpha$.

  Following the scheme of \textbf{Lemma 9.1} of \cite{BS10}
  it's easy to see that
  \begin{align*}
    \mathbb{E}\left\lvert s_j^t Y s_j-\frac{1}{T} tr(Y)\right\rvert^{2r}&=\mathbb{E}\left\lvert \frac{1}{T}\sum_{k,l=1}^T \varepsilon_{j,k-1}y_{kl}\varepsilon_{j,l-1}-\frac{1}{T}\sum_{k=1}^Ty_{kk}\right\rvert^{2r}\\
    &=\mathbb{E}\left\lvert \frac{1}{T}\sum_{k=1}^T(\varepsilon_{j,k-1}^2-1)y_{kk}+\frac{1}{T}\sum_{k\neq l} \varepsilon_{j,k-1}y_{kl}\varepsilon_{j,l-1}\right\rvert^{2r}\\
    &=\mathbb{E}\left\lvert S_1+S_2 \right\rvert^{2r}\leq 2^r
    \frac{\mathbb{E}|S_1|^{2r}+\mathbb{E}|S_2|^{2r}}{2},
  \end{align*}
  where
  \[S_1=\frac{1}{T}\sum_{k=1}^T(\varepsilon^2_{j,k-1}-1)y_{kk},\quad
  S_2=\frac{1}{T}\sum_{1\leq k\neq l\leq
    T}y_{kl}\varepsilon_{j,k-1}\varepsilon_{j,l-1},\] What's more,
  \begin{align*}
    \mathbb{E}|S_1|^{2r}&=\mathbb{E}\left\lvert \frac{1}{T}\sum_{k=1}^T(\varepsilon^2_{j,k-1}-1)y_{kk}\right\rvert^{2r}\\
    &\leq \frac{1}{T^{2r}}\sum_{t=1}^r\sum_{1\leq k_1<\cdots<k_t\leq T}\sum_{{i_1+\cdots+i_t=2r}\atop{i_1\geq2,\cdots,i_t\geq 2}}(2r)!\prod_{l=1}^t\frac{\mathbb{E}(\varepsilon_{j,k_l-1}^2-1)^{i_l}y_{k_lk_l}^{i_l}}{i_l!}\\
    &\leq \frac{1}{T^{2r}}\cdot\frac{1}{v^{2r}}\sum_{t=1}^r T^{t}\sum_{{i_1+\cdots+i_t=2r}\atop{i_1\geq2,\cdots,i_t\geq 2}}\frac{(2r)!}{\prod_{l=1}^ti_l!}\cdot M^t\frac{(\eta T^{\frac{1}{4}})^{4r}}{(\eta T^{\frac{1}{4}})^{4t}}\\
    &\leq \frac{1}{T^{2r}}\cdot\frac{1}{v^{2r}}\sum_{t=1}^r T^{t}
    t^{2r}M^t\frac{(\eta T^{\frac{1}{4}})^{4r}}{(\eta
      T^{\frac{1}{4}})^{4t}}=O(\frac{1}{T^r}),
  \end{align*}
  \[\mathbb{E}|S_2|^{2r}=\frac{1}{T^{2r}}\sum
  y_{i_1j_1}y_{t_1l_1}\cdots
  y_{i_rj_r}y_{t_rl_r}\mathbb{E}(\varepsilon_{j,i_1}\varepsilon_{j,j_1}\varepsilon_{j,t_1}\varepsilon_{j,l_1}\cdots\varepsilon_{j,i_r}\varepsilon_{j,j_r}\varepsilon_{j,t_r}\varepsilon_{j,l_r}).\]
  Consider a graph G with $2r$ edges that link $i_t$ to $j_t$ and
  $l_t$ to $k_t$, $t=1,\cdots,r$. It's easy to see that for any
  nonzero term, the vertex degrees of the graph are not less than
  2. Write the non-coincident vertices as $v_1,\cdots,v_m$ with
  degrees $p_1,\cdots,p_m$ greater than 1, then, similarly in
  \textbf{Lemma 9.1} of \citet{BS10}, we have,
  \[\left\lvert
    \mathbb{E}(\varepsilon_{j,i_1}\varepsilon_{j,j_1}\varepsilon_{j,t_1}\varepsilon_{j,l_1}\cdots\varepsilon_{j,i_r}\varepsilon_{j,j_r}\varepsilon_{j,t_r}\varepsilon_{j,l_r})
  \right\rvert\leq (\eta T^{\frac{1}{4}})^{2(2r-m)},\]
  \[
  \mathbb{E}|S_2|^{2r}\leq \frac{1}{T^{2r}\nu^{2r}}\sum_{m=2}^r
  T^{m/2}(\eta T^{\frac{1}{4}})^{2(2r-m)}m^{4r}=O(\frac{1}{T^r}).
  \]
  Therefore, by the Borel-Cantelli lemma, we have, $\forall 1\leq j\leq
  p$,
  \[s_j^tB_j(\alpha)^{-1}s_j=\frac{1}{T}tr(B_j(\alpha)^{-1})+o_{a.s.}(1),\]
  where the $o_{a.s.}(1)$ terms are uniform in $1\leq j\leq p$.
\end{proof}

\begin{lemma}\label{lem1}
  Under the same assumptions in \textbf{Theorem \ref{th1}}, we have,
  $\forall 1\leq j\leq p$, $1\leq k\leq T-1$, almost surely,
  $$r_j^{t}B_j^{-1}(\alpha)P_1^k r_j=\frac{1}{T}tr\left(B^{-1}\left(\alpha\right)P_{1}^k\right)+o_{a.s.}(1)=o_{a.s.}(1),$$
  $$r_j^t\tilde{C}_jB_j^{-1}(\alpha)P_1^k r_j=\frac{1}{T}tr\left(\tilde{C}B^{-1}\left(\alpha\right)P_{1}^k\right)+o_{a.s.}(1)=o_{a.s.}(1),$$
  where the $o_{a.s.}(1)$ terms are uniform in $1\leq j\leq p$.
\end{lemma}

\begin{proof}

\noindent
Notice that, for $1\leq k\leq T-1$,
\[
P_{1}=\left(\begin{array}{cc}
    \textbf{0} & 0\\
    \textbf{I}_{T-1} & \textbf{0}
  \end{array}\right),
  \quad P_1^k=\left(\begin{array}{cc}
    \textbf{0} & \bf{0}\\
    \textbf{I}_{T-k} & \textbf{0}
  \end{array}\right),\quad
  P_1^T={\bf0},\quad s_j=P_1r_j.
\]
\noindent
Here $P_1^T$ represents the power $T$ of the $T\times T$ matrix $P_1$, we use $P_1^t$ to denote the transpose of matrix $P_1$.
\noindent
Denote, for $1\leq k\leq T$,
\[\frac{1}{T}tr\lb B^{-1}(\alpha)\rb :=x_0,\quad \frac{1}{T}tr\lb B^{-1}(\alpha)C \rb = \frac{1}{T}tr\lb \tilde{C}B^{-1}(\alpha)\rb :=y_0,\]
\[\frac{1}{T}tr\lb B^{-1}(\alpha)P_1^k \rb :=x_k,\quad \frac{1}{T}tr\lb \tilde{C}B^{-1}(\alpha)P_1^k \rb :=y_k.\]
It's easy to see that
\[x_T=y_T=0.\]
\noindent
In addition, for any $1\leq j\leq p$,
\begin{align*}
 s_j^t\tilde{C}_jB_j^{-1}(\alpha)C_jr_j&=s_j^t\tilde{C}_j\lb C_j\tilde{C}_j-\alpha \bf{I}_T\rb^{-1} C_jr_j\\
&=s_j^t\lb {\bf I} -\alpha C_j^{-1}\tilde{C}_j^{-1} \rb^{-1} r_j=s_j^t\tilde{C}_jC_j\lb \tilde{C}_jC_j-\alpha \bf{I}\rb^{-1} r_j\\
&=\alpha \cdot s_j^t \lb \tilde{C}_jC_j-\alpha \bf{I}\rb^{-1} r_j+o_{a.s.}(1)\\
&=\alpha \cdot r_j^t \lb C_j\tilde{C}_j-\alpha \bf{I}\rb^{-1} s_j+o_{a.s.}(1)\\
&=\alpha\frac{1}{T}tr(B^{-1}(\alpha)P_1)+o_{a.s.}(1)=\alpha x_1+o_{a.s.}(1).
\end{align*}
\noindent
Now we can derive the recursion equations between $x_k$ and $y_k$.

\noindent
Firstly, for $x_k$, $1\leq k\leq T-1$, since
\[P_1^k=\left(\sum_{j=1}^ps_js_j^t\sum_{j=1}^pr_jr_j^t\right)B^{-1}(\alpha)P_1^k-\alpha B^{-1}(\alpha)P_1^k,\]
taking trace and dividing $T$ on both sides of the equation, we get
\begin{align*}
 & ~\alpha \cdot \frac{1}{T}tr\lb B^{-1}(\alpha)P_1^k\rb\\
=&\frac{1}{T}\sum_{j=1}^p s_j^t\tilde{C}B^{-1}(\alpha) P_1^ks_j\\
=&\frac{1}{T}\sum_{j=1}^p\dfrac{s_j^t\tilde{C}_j\lb B_j(\alpha)+C_jr_jr_j^t \rb^{-1}P_1^ks_j}{1+s_j^t\tilde{C}_j\lb B_j(\alpha)+C_jr_jr_j^t \rb^{-1}s_j}+o_{a.s.}(1)\\
=&\frac{1}{T}\sum_{j=1}^p\dfrac{1+y_0}{(1+y_0)^2-\alpha x_1^2}\left[s_j^t\tilde{C}_jB_j^{-1}(\alpha)P_1^ks_j-\dfrac{s_j^t\tilde{C}_jB_j^{-1}(\alpha)C_jr_jr_j^tB_j^{-1}(\alpha)P_1^ks_j}{1+r_j^tB^{-1}_j(\alpha)C_jr_j}\right]+o_{a.s.}(1)\\
=&\frac{p}{T}\dfrac{1+y_0}{(1+y_0)^2-\alpha x_1^2}\left[\frac{1}{T}tr(\tilde{C}B^{-1}(\alpha)P_1^k)-\frac{\alpha x_1}{1+y_0}\cdot\frac{1}{T}tr\lb B^{-1}(\alpha)P_1^{k+1}\rb\right]+o_{a.s.}(1),
\end{align*}
i.e.
\begin{equation}\label{rec1}
  \alpha x_k=\frac{p}{T}\cdot\dfrac{1+y_0}{(1+y_0)^2-\alpha x_1^2}\cdot y_k-\frac{p}{T}\cdot\frac{\alpha x_1}{(1+y_0)^2-\alpha x_1^2}\cdot x_{k+1}+o_{a.s.}(1),~1\leq k\leq T-1.
\end{equation}
Particularly, for $k=T-1$, we have
\begin{equation}\label{rec2}
\alpha x_{T-1}=\frac{p}{T}\cdot\dfrac{1+y_0}{(1+y_0)^2-\alpha x_1^2}\cdot y_{T-1}+o_{a.s.}(1).
\end{equation}
\noindent
Similarly, for $y_k$, $1\leq k\leq T$,
\begin{align*}
  y_k&=\frac{1}{T}tr\lb \tilde{C}B^{-1}(\alpha)P_1^k\rb\\
  &=\frac{1}{T}tr\lb \sum_{j=1}^p r_jr_j^t B^{-1}(\alpha)P_1^k\rb=\frac{1}{T}\sum_{j=1}^pr_j^t B^{-1}(\alpha)P_1^k r_j\\
  &=\frac{1}{T}\sum_{j=1}^p \dfrac{r_j^t \lb B_j(\alpha)+s_js_j^t \tilde{C}_j \rb^{-1}P_1^kr_j}{1+r_j^t\lb B_j(\alpha)+s_js_j^t \tilde{C}_j \rb^{-1}C_jr_j}+o_{a.s.}(1)\\
  &=\frac{1}{T}\sum_{j=1}^p\dfrac{1+y_0}{(1+y_0)^2-\alpha x_1^2}\cdot\left[r_j^tB_j^{-1}(\alpha)P_1^kr_j-\dfrac{r_j^tB_j^{-1}(\alpha)s_js_j^t\tilde{C}_jB_j^{-1}(\alpha)P_1^kr_j}{1+s_j^t\tilde{C}_jB^{-1}_j(\alpha)s_j}\right]+o_{a.s.}(1)\\
&=\frac{p}{T}\cdot \dfrac{1+y_0}{(1+y_0)^2-\alpha x_1^2}\cdot\left[\frac{1}{T}tr(B^{-1}(\alpha)P_1^k)-\frac{x_1}{1+y_0}\cdot\frac{1}{T}tr\lb \tilde{C}B^{-1}(\alpha)P_1^{k-1}\rb\right]+o_{a.s.}(1),
\end{align*}
i.e.
\begin{equation}\label{rec3}
  y_k=\frac{p}{T}\cdot \dfrac{1+y_0}{(1+y_0)^2-\alpha x_1^2}\cdot x_k-\frac{p}{T}\cdot \dfrac{x_1}{(1+y_0)^2-\alpha x_1^2}\cdot y_{k-1}+o_{a.s.}(1), ~1\leq k\leq T-1.
\end{equation}

\noindent
Particularly, for $k=T$,
we have
\begin{equation}\label{rec4}
  y_T=\frac{p}{T}\cdot \dfrac{1+y_0}{(1+y_0)^2-\alpha x_1^2}\cdot x_{T}-\frac{p}{T}\cdot \dfrac{x_1}{(1+y_0)^2-\alpha x_1^2}\cdot y_{T-1}+o_{a.s.}(1).
\end{equation}
\noindent
Note that
\[x_T=y_T=0,\]
\noindent
then we have either $x_1=o_{a.s.}(1)$ or $y_{T-1}=o_{a.s.}(1)$.

\noindent
 If
$x_1=o_{a.s.}(1)$, according to equation \eqref{rec1}, we have $y_1=o_{a.s.}(1)$, then according to equation \eqref{rec3}, we have $x_2=y_2=o_{a.s.}(1)$, recursively, we have for all $1\leq k\leq T-1$,
\[x_k=y_k=o_{a.s.}(1).\]
\noindent
Otherwise, if $y_{T-1}=o_{a.s.}(1)$, according to equation \eqref{rec2}, we have $x_{T-1}=o_{a.s.}(1)$, then according to equation \eqref{rec3}, we have $y_{T-2}=o_{a.s.}(1)$, then according to equation \eqref{rec1}, we have $x_{T-2}=o_{a.s.}(1)$, recursively, we still have for all $1\leq k\leq T-1$,
\[x_k=y_k=o_{a.s.}(1).\]
\noindent
Therefore we have,
  $\forall 1\leq j\leq p$, $1\leq k\leq T-1$, almost surely,
  $$r_j^{t}B_j^{-1}(\alpha)P_1^k r_j=\frac{1}{T}tr\left(B^{-1}\left(\alpha\right)P_{1}^k\right)+o_{a.s.}(1)=o_{a.s.}(1),$$
  $$r_j^t\tilde{C}_jB_j^{-1}(\alpha)P_1^k r_j=\frac{1}{T}tr\left(\tilde{C}B^{-1}\left(\alpha\right)P_{1}^k\right)+o_{a.s.}(1)=o_{a.s.}(1),$$
  where the $o_{a.s.}(1)$ terms are uniform in $1\leq j\leq p$.
\end{proof}

\begin{lemma}\label{lem2}
Extension of Lemma \ref{lem1} to time lag $\tau$:

 we have,
  $\forall 1\leq j\leq p$, $1\leq k\leq [\frac{T}{\tau}]$, almost surely,
  $$r_j^{t}B_j^{-1}(\alpha)(P_1^{\tau})^k r_j=\frac{1}{T}tr\left(B^{-1}\left(\alpha\right)(P_1^{\tau})^k\right)+o_{a.s.}(1)=o_{a.s.}(1),$$
  $$r_j^t\tilde{C}_jB_j^{-1}(\alpha)(P_1^{\tau})^k r_j=\frac{1}{T}tr\left(\tilde{C}B^{-1}\left(\alpha\right)(P_1^{\tau})^k\right)+o_{a.s.}(1)=o_{a.s.}(1),$$
  where the $o_{a.s.}(1)$ terms are uniform in $1\leq j\leq p$.
\end{lemma}

\begin{proof}

\noindent

Denote, for $1\leq k\leq \left[\frac{T}{\tau}\right]$,
\[\frac{1}{T}tr\lb B^{-1}(\alpha)\rb :=x_0,\quad \frac{1}{T}tr\lb B^{-1}(\alpha)C \rb = \frac{1}{T}tr\lb \tilde{C}B^{-1}(\alpha)\rb :=y_0,\]
\[\frac{1}{T}tr\lb B^{-1}(\alpha)(P_1^{\tau})^k \rb :=x_k,\quad \frac{1}{T}tr\lb \tilde{C}B^{-1}(\alpha)(P_1^{\tau})^k \rb :=y_k.\]
It's easy to see that
\[x_{\left[\frac{T}{\tau}\right]+1}=y_{\left[\frac{T}{\tau}\right]+1}=0.\]
\noindent
In addition, for any $1\leq j\leq p$,
\[
 s_j^t\tilde{C}_jB_j^{-1}(\alpha)C_jr_j=\alpha\frac{1}{T}tr(B^{-1}(\alpha)P_1^{\tau})+o_{a.s.}(1)=\alpha x_1+o_{a.s.}(1).
\]
\noindent
Now we can derive the recursion equations between $x_k$ and $y_k$.

\noindent
Firstly, for $x_k$, $1\leq k\leq \left[\frac{T}{\tau}\right]$,
\begin{align*}
  \alpha &\cdot \frac{1}{T}tr\lb B^{-1}(\alpha)(P_1^{\tau})^k\rb =o_{a.s.}(1)+\\
  &~\frac{p}{T}\dfrac{1+y_0}{(1+y_0)^2-\alpha x_1^2}\left[\frac{1}{T}tr(\tilde{C}B^{-1}(\alpha)(P_1^{\tau})^k)-\frac{\alpha x_1}{1+y_0}\cdot\frac{1}{T}tr\lb B^{-1}(\alpha)(P_1^{\tau})^{k+1}\rb\right],
\end{align*}
i.e.
\begin{equation}\label{rec5}
  \alpha x_k=\frac{p}{T}\cdot\dfrac{1+y_0}{(1+y_0)^2-\alpha x_1^2}\cdot y_k-\frac{p}{T}\cdot\frac{\alpha x_1}{(1+y_0)^2-\alpha x_1^2}\cdot x_{k+1}+o_{a.s.}(1),~1\leq k\leq \left[\frac{T}{\tau}\right].
\end{equation}

Similarly, for $y_k$, $1\leq k\leq \left[\frac{T}{\tau}\right]+1$,
\begin{align*}
  y_k&=\frac{1}{T}tr\lb \tilde{C}B^{-1}(\alpha)(P_1^{\tau})^k \rb\\
  &=\frac{p}{T}\cdot \dfrac{1+y_0}{(1+y_0)^2-\alpha x_1^2}\cdot\left[\frac{1}{T}tr(B^{-1}(\alpha)(P_1^{\tau})^k)-\frac{x_1}{1+y_0}\cdot\frac{1}{T}tr\lb \tilde{C}B^{-1}(\alpha)(P_1^{\tau})^{k-1}\rb\right]+o_{a.s.}(1),
\end{align*}
i.e.
\begin{equation}\label{rec6}
  y_k=\frac{p}{T}\cdot \dfrac{1+y_0}{(1+y_0)^2-\alpha x_1^2}\cdot x_k-\frac{p}{T}\cdot \dfrac{x_1}{(1+y_0)^2-\alpha x_1^2}\cdot y_{k-1}+o_{a.s.}(1), ~1\leq k\leq \left[\frac{T}{\tau}\right]+1.
\end{equation}

\noindent
Particularly, for $k=\left[\frac{T}{\tau}\right]+1$,
we have
\begin{equation}\label{rec4}
  y_{\left[\frac{T}{\tau}\right]+1}=\frac{p}{T}\cdot \dfrac{1+y_0}{(1+y_0)^2-\alpha x_1^2}\cdot x_{\left[\frac{T}{\tau}\right]+1}-\frac{p}{T}\cdot \dfrac{x_1}{(1+y_0)^2-\alpha x_1^2}\cdot y_{\left[\frac{T}{\tau}\right]}+o_{a.s.}(1).
\end{equation}
\noindent
Note that
\[x_{\left[\frac{T}{\tau}\right]+1}=y_{\left[\frac{T}{\tau}\right]+1}=0,\]
\noindent
following the same arguments in Lemma \ref{lem1},  we have,
  $\forall 1\leq j\leq p$, $1\leq k\leq \left[\frac{T}{\tau}\right]$, almost surely,
  $$r_j^{t}B_j^{-1}(\alpha)(P_1^{\tau})^k r_j=\frac{1}{T}tr\left(B^{-1}\left(\alpha\right)(P_1^{\tau})^k\right)+o_{a.s.}(1)=o_{a.s.}(1),$$
  $$r_j^t\tilde{C}_jB_j^{-1}(\alpha)(P_1^{\tau})^k r_j=\frac{1}{T}tr\left(\tilde{C}B^{-1}\left(\alpha\right)(P_1^{\tau})^k\right)+o_{a.s.}(1)=o_{a.s.}(1),$$
  where the $o_{a.s.}(1)$ terms are uniform in $1\leq j\leq p$.

\end{proof}

\appendix

\section{Justification of truncation, centralization and
  standardization}\label{app}

Recall that
$\varepsilon_{t}=\left(\varepsilon_{1t},\cdots,\varepsilon_{pt}\right)^{t}$,
$\varepsilon_{it}$ are independent real-valued random variables with
$\mathbb{E}\left(\varepsilon_{it}\right)=0,\mathbb{E}\left(|\varepsilon_{it}|^{2}\right)=1$,
and we are interested in is the LSD of time-lagged covariance matrix

\[
A=\frac{1}{T^{2}}\left(\sum_{i=1}^{T}\varepsilon_{i}\varepsilon_{i-1}^{t}\right)\left(\sum_{j=1}^{T}\varepsilon_{j-1}\varepsilon_{j}^{t}\right).
\]

The assumed moment conditions are: for some constant $M$,

\[
\sup_{1\leq i\leq p,0\leq t\leq
  T}\mathbb{E}\left(|\varepsilon_{it}|^{4}\right)<M,
\]

and for any $\eta>0$,

\[
\frac{1}{\eta^{4}pT}\sum_{i=1}^{p}\sum_{t=0}^{T}\mathbb{E}\left(|\varepsilon_{it}|^{4}I_{(|\varepsilon_{it}|\geq\eta
    T^{1/4})}\right)=o\left(1\right).
\]

The aim of the truncation, centralization and standardization
procedure is that after these treatment, we may assume that

\[
|\varepsilon_{ij}|\leq\eta T^{1/4},\quad
\mathbb{E}\left(\varepsilon_{ij}\right)=0,\quad
Var\left(\varepsilon_{ij}\right)=1,\quad
\mathbb{E}\left(|\varepsilon_{ij}|^{4}\right)<M.
\]

Since the whole procedure is the same with respect to different time lag $\tau$, we focus on the case of lag-1 sample auto-covariance matrix.

\subsection{Truncation}

Let
$\tilde{\varepsilon}_{jt}=\varepsilon_{jt}I_{(|\varepsilon_{jt}|<\eta
  T^{1/4})}$,
$\tilde{\varepsilon}_{t}=\left(\tilde{\varepsilon}_{1t},\cdots,\tilde{\varepsilon}_{pt}\right)^{t}$,
$\eta$ can be seen as a constant.

Define
$$\tilde{A}=\frac{1}{T^{2}}\left(\sum_{i=1}^{T}\tilde{\varepsilon}_{i}\tilde{\varepsilon}_{i-1}^{t}\right)\left(\sum_{j=1}^{T}\tilde{\varepsilon}_{j-1}\tilde{\varepsilon}_{j}^{t}\right),$$
then according to Theorem A.44 of \cite{BS10}
which states that

\[
\lVert F^{AA^{*}}-F^{BB^{*}}\rVert \leq\frac{1}{p}
\mathrm{rank}\left(A-B\right),
\]

\noindent
we have
\begin{align*}
  \lVert F^{A}-F^{\tilde{A}}\rVert & \leq  \frac{1}{p} \mathrm{rank}\left(\frac{1}{T}\sum_{i=1}^{T}\tilde{\varepsilon}_{i}\tilde{\varepsilon}_{i-1}^{t}-\frac{1}{T}\sum_{i=1}^{T}\varepsilon_{i}\varepsilon_{i-1}^{t}\right)\\
  &\leq \frac{1}{p}
  \mathrm{rank}\left(\frac{1}{T}\sum_{i=1}^{T}\tilde{\varepsilon}_{i}(\tilde{\varepsilon}_{i-1}^{t}-\varepsilon_{i-1}^{t})\right)
  +\frac{1}{p} \mathrm{rank}\left(\frac{1}{T}\sum_{i=1}^{T}(\tilde{\varepsilon}_{i}-\varepsilon_{i})\varepsilon_{i-1}^{t}\right)\\
  &\leq \frac{1}{p}\sum_{i=1}^{T}
  \mathrm{rank}\left(\frac{1}{T}\tilde{\varepsilon}_{i}(\tilde{\varepsilon}_{i-1}^{t}-\varepsilon_{i-1}^{t})\right)
  +\frac{1}{p}\sum_{i=1}^{T} \mathrm{rank}\left(\frac{1}{T}(\tilde{\varepsilon}_{i}-\varepsilon_{i})\varepsilon_{i-1}^{t}\right)\\
  & \leq
  \frac{2}{p}\sum_{t=0}^{T}\sum_{i=1}^{p}I_{(|\varepsilon_{it}|\geq\eta
    T^{1/4})},
\end{align*}

\begin{align*}
  \mathbb{E}\left(\frac{1}{p}\sum_{t=0}^{T}\sum_{i=1}^{p}I_{(|\varepsilon_{it}|\geq\eta
      T^{1/4})}\right)
  &\leq \frac{1}{p}\sum_{t=0}^{T}\sum_{i=1}^{p}\mathbb{E}\left(\frac{|\varepsilon_{it}|^{4}}{\eta^{4}\cdot T}I_{(|\varepsilon_{it}|\geq\eta T^{1/4})}\right)\\
  & =
  \frac{1}{\eta^{4}pT}\sum_{i=1}^{p}\sum_{t=0}^{T}\mathbb{E}\left(|\varepsilon_{it}|^{4}I_{(|\varepsilon_{it}|\geq\eta
      T^{1/4})}\right) =o\left(1\right),
\end{align*}

\begin{align*}
  Var\left(\frac{1}{p}\sum_{t=0}^{T}\sum_{i=1}^{p}I_{(|\varepsilon_{it}|\geq\eta
      T^{1/4})}\right)
  &=  \frac{1}{p^{2}}\sum_{t=0}^{T}\sum_{i=1}^{p}Var\left(I_{(|\varepsilon_{it}|\geq\eta T^{1/4})}\right)\\
  &\leq
  \frac{1}{p^{2}}\sum_{t=0}^{T}\sum_{i=1}^{p}\mathbb{E}\left(I_{(|\varepsilon_{it}|\geq\eta
      T^{1/4})}\right) =o\left(\frac{1}{T}\right).
\end{align*}

Applying Bernstein's inequality

\[
\mathbb{P}\left(|S_{n}|\geq\varepsilon\right)\leq2\exp\left(-\frac{\varepsilon^{2}}{2\left(B_{n}^{2}+b\varepsilon\right)}\right),
\]

\noindent
where $S_{n}=\sum_{i=1}^{n}X_{i}$, $B_{n}^2=\mathbb{E}S_{n}^{2}$,
$X_{i}$ are i.i.d bounded by b, we can get that, for any small
$\varepsilon>0$,
\[
\mathbb{P}\left(\frac{1}{p}\sum_{t=0}^{T}\sum_{i=1}^{p}I_{(|\varepsilon_{it}|\geq\eta
    T^{1/4})}\geq\varepsilon\right)\leq2\exp\left(-\frac{\varepsilon^{2}}{2\left(\frac{\varepsilon}{p}+o\left(\frac{1}{T}\right)\right)}\right)=2\exp\left(-K_{\varepsilon}p\right),
\]
\noindent
which is summable, then by Borel-Cantelli lemma,

\[
a.s.\lVert F^{A}-F^{\tilde{A}}\rVert \rightarrow0, \mbox{as
}T\rightarrow\infty.
\]

\subsection{Centralization}

Let
$\hat{\varepsilon}_{it}=\tilde{\varepsilon}_{it}-\mathbb{E}\left(\tilde{\varepsilon}_{it}\right)$,
$~\hat{\varepsilon}_{t}=\left(\hat{\varepsilon}_{1t},\cdots,\hat{\varepsilon}_{pt}\right)$,
$~\hat{A}=\frac{1}{T^{2}}\left(\sum_{i=1}^{T}\hat{\varepsilon}_{i}\hat{\varepsilon}_{i-1}^{t}\right)\left(\sum_{j=1}^{T}\hat{\varepsilon}_{j-1}\hat{\varepsilon}_{j}^{t}\right)$.

With Theorem A.46 of \cite{BS10},
\[
L^{4}\left(F^{AA^{*}},F^{BB^{*}}\right)\leq\frac{2}{p^{2}}tr\left(AA^{*}+BB^{*}\right)tr\left(\left(A-B\right)\left(A-B\right)^{*}\right),
\]

we have

\begin{align*}
  L^{4}\left(F^{\hat{A}},F^{\tilde{A}}\right) & \leq \frac{2}{p^{2}}tr\left(\frac{1}{T^{2}}\left(\sum_{i=1}^{T}\hat{\varepsilon}_{i}\hat{\varepsilon}_{i-1}^{t}\right)\left(\sum_{j=1}^{T}\hat{\varepsilon}_{j-1}\hat{\varepsilon}_{j}^{t}\right)+\frac{1}{T^{2}}\left(\sum_{i=1}^{T}\tilde{\varepsilon}_{i}\tilde{\varepsilon}_{i-1}^{t}\right)\left(\sum_{j=1}^{T}\tilde{\varepsilon}_{j-1}\tilde{\varepsilon}_{j}^{t}\right)\right)\\
  & \quad \cdot tr\left(\frac{1}{T^{2}}\left(\sum_{i=1}^{T}\hat{\varepsilon}_{i}\hat{\varepsilon}_{i-1}^{t}-\sum_{i=1}^{T}\tilde{\varepsilon}_{i}\tilde{\varepsilon}_{i-1}^{t}\right)\left(\sum_{j=1}^{T}\hat{\varepsilon}_{j-1}\hat{\varepsilon}_{j}^{t}-\sum_{j=1}^{T}\tilde{\varepsilon}_{j-1}\tilde{\varepsilon}_{j}^{t}\right)\right)\\
  & := N_1\cdot N_2.
\end{align*}
\noindent
For $N_2$,
\begin{align}\label{eq12}
  N_2 & = tr\left(\frac{1}{T^{2}}\left(\sum_{i=1}^{T}\hat{\varepsilon}_{i}\hat{\varepsilon}_{i-1}^{t}-\sum_{i=1}^{T}\tilde{\varepsilon}_{i}\tilde{\varepsilon}_{i-1}^{t}\right)\left(\sum_{j=1}^{T}\hat{\varepsilon}_{j-1}\hat{\varepsilon}_{j}^{t}-\sum_{j=1}^{T}\tilde{\varepsilon}_{j-1}\tilde{\varepsilon}_{j}^{t}\right)\right)\nonumber\\
  &=
  tr\left(\frac{1}{T^{2}}\sum_{i=1}^{T}\left(\mathbb{E}\left(\tilde{\varepsilon}_{i}\right)\mathbb{E}\left(\tilde{\varepsilon}_{i-1}^{t}\right)-\mathbb{E}\left(\tilde{\varepsilon}_{i}\right)\tilde{\varepsilon}_{i-1}^{t}-\tilde{\varepsilon}_{i}\mathbb{E}\left(\tilde{\varepsilon}_{i-1}^{t}\right)\right)\right.\nonumber
  \\
  &\quad \cdot \left.\sum_{i=1}^{T}\left(\mathbb{E}\left(\tilde{\varepsilon}_{i}\right)\mathbb{E}\left(\tilde{\varepsilon}_{i-1}^{t}\right)-\mathbb{E}\left(\tilde{\varepsilon}_{i}\right)\tilde{\varepsilon}_{i-1}^{t}-\tilde{\varepsilon}_{i}\mathbb{E}\left(\tilde{\varepsilon}_{i-1}^{t}\right)\right)\right)^{t}\nonumber\\
  &=  \left\lVert \frac{1}{T}\sum_{i=1}^{T}\left(\mathbb{E}\left(\tilde{\varepsilon}_{i}\right)\mathbb{E}\left(\tilde{\varepsilon}_{i-1}^{t}\right)-\mathbb{E}\left(\tilde{\varepsilon}_{i}\right)\tilde{\varepsilon}_{i-1}^{t}-\tilde{\varepsilon}_{i}\mathbb{E}\left(\tilde{\varepsilon}_{i-1}^{t}\right)\right)\right\rVert ^{2}\nonumber\\
  &\leq 2\left\lVert
    \frac{1}{T}\sum_{i=1}^{T}\mathbb{E}\left(\tilde{\varepsilon}_{i}\right)\mathbb{E}\left(\tilde{\varepsilon}_{i-1}^{t}\right)\right\rVert^{2}+2\left\lVert
    \frac{1}{T}\sum_{i=1}^{T}\mathbb{E}\left(\tilde{\varepsilon}_{i}\right)\tilde{\varepsilon}_{i-1}^{t}\right\rVert^{2}+2\left\lVert
    \frac{1}{T}\sum_{i=1}^{T}\tilde{\varepsilon}_{i}\mathbb{E}\left(\tilde{\varepsilon}_{i-1}^{t}\right)\right\rVert^{2}.
\end{align}
\noindent
Consider the second term, we have
\begin{align*}
  & \left\lVert \frac{1}{T}\sum_{i=1}^{T}\mathbb{E}\left(\tilde{\varepsilon}_{i}\right)\tilde{\varepsilon}_{i-1}^{t}\right\rVert^{2}=\frac{1}{T^{2}}\sum_{i,j=1}^{p}\left(\sum_{t=1}^{T}\tilde{\varepsilon}_{j,t-1}\mathbb{E}\left(\tilde{\varepsilon}_{it}\right)\right)^{2}\\
  =& \frac{1}{T^{2}}\sum_{i,j=1}^{p}\sum_{t=1}^{T}\tilde{\varepsilon}_{j,t-1}^{2}\left(\mathbb{E}\left(\tilde{\varepsilon}_{it}\right)\right)^{2}+\frac{1}{T^{2}}\sum_{i,j=1}^{p}\sum_{t_{1}\neq t_{2}}\tilde{\varepsilon}_{j,t_{1}-1}\tilde{\varepsilon}_{j,t_{2}-1}\mathbb{E}\left(\tilde{\varepsilon}_{it_{1}}\right)\mathbb{E}\left(\tilde{\varepsilon}_{it_{2}}\right)\\
  {=:} & M_{1}+M_{2}.
\end{align*}

Notice that $\sup_{1\leq i\leq p,1\leq t\leq
  T}\mathbb{E}\left(\varepsilon_{it}^{4}\right)<M$, we have

\begin{align*}
  \mathbb{E}\left(M_{1}\right) & =  \frac{1}{T^{2}}\sum_{i,j=1}^{p}\sum_{t=1}^{T}\mathbb{E}\left(\tilde{\varepsilon}_{j,t-1}^{2}\right)\left(\mathbb{E}\left(\tilde{\varepsilon}_{it}\right)\right)^{2}\\
  & \leq \frac{C_{1}}{T^{2}}\sum_{i,j=1}^{p}\sum_{t=1}^{T}\left(\mathbb{E}\left(|\varepsilon_{it}|I_{(|\varepsilon_{it}|\geq\eta T^{1/4})}\right)\right)^{2}\\
  & \leq  \frac{C_{1}}{T^{2}}\sum_{i,j=1}^{p}\sum_{t=1}^{T}\frac{1}{\eta^{6}\cdot T^{3/2}}\left(\mathbb{E}\left(|\varepsilon_{it}|^{4}I_{(|\varepsilon_{it}|\geq\eta T^{1/4})}\right)\right)^{2}\\
  & = O\left(T^{-\frac{1}{2}}\right),
\end{align*}
\noindent
Moreover,
\begin{align*}
  Var\left(M_{1}\right) & =  \frac{1}{T^{4}}\sum_{j=1}^{p}\sum_{t=1}^{T}\mathbb{E}\left(\tilde{\varepsilon}_{j,t-1}^{2}-\mathbb{E}\left(\tilde{\varepsilon}_{j,t-1}^{2}\right)\right)^{2}\left(\sum_{i=1}^{p}\left(\mathbb{E}\left(\tilde{\varepsilon}_{it}\right)\right)^{2}\right)^{2}\\
  & \leq  \frac{1}{T^{4}}\sum_{j=1}^{p}\sum_{t=1}^{T}\mathbb{E}\left(\tilde{\varepsilon}_{j,t-1}^{2}\right)^{4}\left(\sum_{i=1}^{p}\left(\mathbb{E}\left(|\varepsilon_{it}|I_{(|\varepsilon_{it}|\geq\eta\cdot T^{1/4})}\right)\right)^{2}\right)^{2}\\
  & \leq
  \frac{C_{2}}{T^{4}}\sum_{j=1}^{p}\sum_{t=1}^{T}\frac{1}{T^{3}}\left(\sum_{i=1}^{p}\left(\mathbb{E}\left(|\varepsilon_{it}|^{4}I_{(|\varepsilon_{it}|\geq\eta\cdot
          T^{1/4})}\right)\right)^{2}\right)^{2}=O\left(T^{-3}\right).
\end{align*}
\noindent
Therefore, $a.s. \quad M_{1}\rightarrow0, \mbox{as }
T\rightarrow\infty$.

\noindent
For the term $M_2$, we have
\begin{align*}
  \mathbb{E}\left(M_{2}\right) & =  \frac{1}{T^{2}}\sum_{i,j=1}^{p}\sum_{t_{1}\neq t_{2}}\mathbb{E}\left(\tilde{\varepsilon}_{j,t_{1}-1}\tilde{\varepsilon}_{j,t_{2}-1}\right)\mathbb{E}\left(\tilde{\varepsilon}_{it_{1}}\right)\mathbb{E}\left(\tilde{\varepsilon}_{it_{2}}\right)\\
  & =  \frac{1}{T^{2}}\sum_{i,j=1}^{p}\sum_{t_{1}\neq t_{2}}\mathbb{E}\left(\tilde{\varepsilon}_{j,t_{1}-1}\right)\mathbb{E}\left(\tilde{\varepsilon}_{j,t_{2}-1}\right)\mathbb{E}\left(\tilde{\varepsilon}_{it_{1}}\right)\mathbb{E}\left(\tilde{\varepsilon}_{it_{2}}\right)\\
  & \leq \frac{1}{T^{2}}\sum_{i,j=1}^{p}\sum_{t_{1}\neq
    t_{2}}\frac{1}{\eta^{12}\cdot T^{3}}\left(\sup_{1\leq i\leq
      p,0\leq t\leq
      T}\mathbb{E}\left(|\varepsilon_{it}|^{4}I_{(|\varepsilon_{it}|\geq\eta\cdot
        T^{1/4})}\right)\right)^{4}=O\left(T^{-1}\right),
\end{align*}

\begin{align*}
  Var\left(M_{2}\right) & =  \frac{1}{T^{4}}\sum_{j=1}^{p}\sum_{t_{1}\neq t_{2}}Var\left(\tilde{\varepsilon}_{j,t_{1}-1}\tilde{\varepsilon}_{j,t_{2}-1}\right)\left(\sum_{i=1}^{p}\mathbb{E}\left(\tilde{\varepsilon}_{it_{1}}\right)\mathbb{E}\left(\tilde{\varepsilon}_{it_{2}}\right)\right)^{2}\\
  & \leq  \frac{1}{T^{4}}\sum_{j=1}^{p}\sum_{t_{1}\neq t_{2}}\mathbb{E}\left(\tilde{\varepsilon}_{j,t_{1}-1}^{2}\right)\mathbb{E}\left(\tilde{\varepsilon}_{j,t_{2}-1}^{2}\right)\left(\sum_{i=1}^{p}\left(\sup_{1\leq i\leq p,0\leq t\leq T}\mathbb{E}\left(\tilde{\varepsilon}_{it}\right)\right)^{2}\right)^{2}\\
  & \leq \frac{C_{3}}{T^{4}}\sum_{j=1}^{p}\sum_{t_{1}\neq
    t_{2}}\frac{1}{T^{3}}\left(\sum_{i=1}^{p}\left(\sup_{1\leq i\leq
        p,0\leq t\leq
        T}\mathbb{E}\left(|\varepsilon_{it}|^{4}I_{(|\varepsilon_{it}|\geq\eta\cdot
          T^{1/4})}\right)\right)^{2}\right)^{2}=O\left(T^{-2}\right).
\end{align*}
\noindent
Therefore, a.s. $M_{2}\rightarrow0$, as $T\rightarrow\infty$.

Consequently, $\left\lVert
  \frac{1}{T}\sum_{i=1}^{T}\mathbb{E}\left(\tilde{\varepsilon}_{i}\right)\tilde{\varepsilon}_{i-1}^{t}\right\rVert^{2}\rightarrow
0, a.s.$ Similarly, we can prove that the last term in equation
\eqref{eq12} tends to zero almost surely. As for the first term, we have
\begin{align*}
  \left\lVert\frac{1}{T}\sum_{i=1}^{T}\mathbb{E}\left(\tilde{\varepsilon}_{i}\right)\mathbb{E}\left(\tilde{\varepsilon}_{i-1}^{t}\right)\right\rVert^{2} & = \sum_{i,j=1}^{p}\left(\frac{1}{T}\sum_{t=1}^{T}\left(\mathbb{E}\left(\tilde{\varepsilon}_{it}\right)\mathbb{E}\left(\tilde{\varepsilon}_{j,t-1}\right)\right)\right)^{2}\\
  & = \frac{1}{T^{2}}\sum_{i,j=1}^{p}\sum_{t_{1}=1}^{T}\sum_{t_{2}=1}^{T}\mathbb{E}\left(\tilde{\varepsilon}_{it_{1}}\right)\mathbb{E}\left(\tilde{\varepsilon}_{j,t_{1}-1}\right)\mathbb{E}\left(\tilde{\varepsilon}_{it_{2}}\right)\mathbb{E}\left(\tilde{\varepsilon}_{j,t_{2}-1}\right)\\
  & \leq
  \frac{C_{4}}{T^{2}}\sum_{i,j=1}^{p}\sum_{t_{1}=1}^{T}\sum_{t_{2}=1}^{T}\frac{1}{T^{3}}\left(\sup_{1\leq
      i\leq p,0\leq t\leq
      T}\mathbb{E}\left(|\varepsilon_{it}|^{4}I_{(|\varepsilon_{it}|\geq\eta\cdot
        T^{1/4})}\right)\right)^{4}=O\left(T^{-1}\right).
\end{align*}
\noindent
Therefore
\[
N_1=tr\left(\frac{1}{T^{2}}\left(\sum_{i=1}^{T}\hat{\varepsilon}_{i}\hat{\varepsilon}_{i-1}^{t}-\sum_{i=1}^{T}\tilde{\varepsilon}_{i}\tilde{\varepsilon}_{i-1}^{t}\right)\left(\sum_{j=1}^{T}\hat{\varepsilon}_{j-1}\hat{\varepsilon}_{j}^{t}-\sum_{j=1}^{T}\tilde{\varepsilon}_{j-1}\tilde{\varepsilon}_{j}^{t}\right)\right)\rightarrow0,
a.s.
\]
\noindent
Now, we consider $N_1$,
\[
\frac{1}{p^{2}}tr\left(\frac{1}{T^{2}}\left(\sum_{i=1}^{T}\hat{\varepsilon}_{i}\hat{\varepsilon}_{i-1}^{t}\right)\left(\sum_{j=1}^{T}\hat{\varepsilon}_{j-1}\hat{\varepsilon}_{j}^{t}\right)+\frac{1}{T^{2}}\left(\sum_{i=1}^{T}\tilde{\varepsilon}_{i}\tilde{\varepsilon}_{i-1}^{t}\right)\left(\sum_{j=1}^{T}\tilde{\varepsilon}_{j-1}\tilde{\varepsilon}_{j}^{t}\right)\right){=:}M_{3}+M_{4},
\]
\noindent
Firstly, for $M_{3}$, since
$\mathbb{E}\left(\hat{\varepsilon}_{it}\right)=0$,
\begin{align*}
  \mathbb{E}\left(M_{3}\right) & =  \mathbb{E}\left(\frac{1}{p^{2}T^{2}}\sum_{i,j=1}^{p}\left(\sum_{t=1}^{T}\hat{\varepsilon}_{it}\hat{\varepsilon}_{j,t-1}\right)^{2}\right)\\
  & =
  \frac{1}{p^{2}T^{2}}\sum_{i,j=1}^{p}\sum_{t=1}^{T}\mathbb{E}\left(\hat{\varepsilon}_{it}^{2}\right)\mathbb{E}\left(\hat{\varepsilon}_{j,t-1}^{2}\right)=O\left(\frac{1}{T}\right).
\end{align*}
\noindent
Moreover,
\begin{align*}
  Var\left(M_{3}\right) & =  \mathbb{E}\left(\frac{1}{p^{2}T^{2}}\sum_{i,j=1}^{p}\left(\sum_{t=1}^{T}\hat{\varepsilon}_{it}\hat{\varepsilon}_{j,t-1}\right)^{2}\right)^{2}-\left(\mathbb{E}\left(M_{3}\right)\right)^{2}\\
  & = \frac{1}{p^{4}T^{4}}\mathbb{E}\left(\sum_{i,j=1}^{p}\sum_{t=1}^{T}\hat{\varepsilon}_{it}^{2}\hat{\varepsilon}_{j,t-1}^{2}\right)^{2}+\frac{1}{p^{4}T^{4}}\mathbb{E}\left(\sum_{i,j=1}^{p}\sum_{t_{1}\neq t_{2}}\hat{\varepsilon}_{it_{1}}\hat{\varepsilon}_{j,t_{1}-1}\hat{\epsilon}_{it_{2}}\hat{\varepsilon}_{j,t_{2}-1}\right)^{2}+O\left(\frac{1}{T^{2}}\right)\\
  & \leq
  O\left(\frac{1}{T^{2}}\right)+O\left(\frac{1}{T^{3}}\right)+O\left(\frac{1}{T^{2}}\right)=O\left(\frac{1}{T^{2}}\right).
\end{align*}
\noindent
Therefore $M_3\gt 0$, a.s. Next for $M_4$,

\begin{align*}
  \mathbb{E}\left(M_{4}\right) & = \mathbb{E}\left(\frac{1}{p^{2}T^{2}}\sum_{i,j=1}^{p}\left(\sum_{t=1}^{T}\tilde{\varepsilon}_{it}\tilde{\varepsilon}_{j,t-1}\right)^{2}\right)\\
  & =  \frac{1}{p^{2}T^{2}}\sum_{i,j=1}^{p}\sum_{t=1}^{T}\mathbb{E}\tilde{\varepsilon}_{it}^{2}\mathbb{E}\tilde{\varepsilon}_{j,t-1}^{2}+\frac{1}{p^{2}T^{2}}\sum_{i,j=1}^{p}\sum_{t_{1}\neq t_{2}}\mathbb{E}\left(\tilde{\varepsilon}_{it_{1}}\right)\mathbb{E}\left(\tilde{\varepsilon}_{j,t_{1}-1}\right)\mathbb{E}\left(\tilde{\varepsilon}_{it_{2}}\right)\mathbb{E}\left(\tilde{\varepsilon}_{j,t_{2}-1}\right)\\
  & \leq
  O\left(\frac{1}{T}\right)+\frac{1}{p^{2}T^{2}}\sum_{i,j=1}^{p}\sum_{t_{1}\neq
    t_{2}}\frac{1}{\eta^{12}T^{3}}\left(\sup_{1\leq i\leq p,0\leq
      t\leq
      T}\mathbb{E}\left(|\varepsilon_{it}|^{4}I_{(|\varepsilon_{it}|\geq\eta\cdot
        T^{1/4})}\right)\right)^{4}=O\left(\frac{1}{T}\right).
\end{align*}

\begin{align*}
  Var\left(M_{4}\right) & = \frac{1}{p^{4}T^{4}}Var\left(\sum_{i,j=1}^{p}\left(\sum_{t=1}^{T}\tilde{\varepsilon}_{it}\tilde{\varepsilon}_{j,t-1}\right)^{2}\right)\\
  & \leq  \frac{1}{p^{4}T^{4}}\mathbb{E}\left(\sum_{i,j=1}^{p}\left(\sum_{t=1}^{T}\tilde{\varepsilon}_{it}\tilde{\varepsilon}_{j,t-1}\right)^{2}\right)^{2}\\
  & =  \frac{1}{p^{4}T^{4}}\mathbb{E}\left(\sum_{i,j=1}^{p}\sum_{t=1}^{T}\tilde{\varepsilon}_{it}^{2}\tilde{\varepsilon}_{j,t-1}^{2}\right)^{2}+\frac{1}{p^{4}T^{4}}\mathbb{E}\left(\sum_{i,j=1}^{p}\sum_{t_{1}\neq t_{2}}\tilde{\varepsilon}_{it_{1}}\tilde{\varepsilon}_{j,t_{1}-1}\tilde{\varepsilon}_{it_{2}}\tilde{\varepsilon}_{j,t_{2}-1}\right)^{2}\\
  & \leq
  O\left(\frac{1}{T^{2}}\right)+O\left(\frac{1}{T^{6}}\right)=O\left(\frac{1}{T^{2}}\right).
\end{align*}

Therefore, $M_{4}\rightarrow0, a.s. $. All in all,
\[
L^{4}\left(F^{\hat{A}},F^{\tilde{A}}\right)\leq N_1\cdot N_2
\leq4\left(M_{3}+M_{4}\right)\left(M_{1}+M_{2}\right)\rightarrow 0,
a.s. T\rightarrow\infty.
\]

\subsection{Rescaling}

Define
$\hat{\sigma}_{ij}^{2}=\mathbb{E}|\hat{\varepsilon}_{ij}|^{2}=\mathbb{E}|\tilde{\varepsilon}_{ij}-\mathbb{E}\tilde{\varepsilon}_{ij}|^{2}$,
we can see that as $T\rightarrow\infty$,
$\hat{\sigma}_{ij}^{2}\rightarrow1$ since
$\mathbb{E}(\varepsilon_{ij})=0$,
$Var\left(\varepsilon_{ij}\right)=1$.

According to Theorem A.46 of \cite{BS10}, we have
\begin{align*}
  L^{4}\left(F^{\hat{A}},F^{\hat{\sigma}_{ij}^{-4}\hat{A}}\right) & \leq  \frac{2}{p^{2}}\left[\frac{1+\hat{\sigma}_{ij}^{-4}}{T^{2}}tr\left(\left(\sum_{i=1}^{T}\hat{\varepsilon}_{i}\hat{\varepsilon}_{i-1}^{t}\right)\left(\sum_{j=1}^{T}\hat{\varepsilon}_{j-1}\hat{\varepsilon}_{j}^{t}\right)\right)\right]\\
  &\cdot\left[\frac{1-\hat{\sigma}_{ij}^{-4}}{T^{2}}tr\left(\left(\sum_{i=1}^{T}\hat{\varepsilon}_{i}\hat{\varepsilon}_{i-1}^{t}\right)\left(\sum_{j=1}^{T}\hat{\varepsilon}_{j-1}\hat{\varepsilon}_{j}^{t}\right)\right)\right]\\
  & =
  2\left(1-\hat{\sigma}_{ij}^{-8}\right)\left[\frac{1}{pT^{2}}tr\left(\left(\sum_{i=1}^{T}\hat{\varepsilon}_{i}\hat{\varepsilon}_{i-1}^{t}\right)\left(\sum_{j=1}^{T}\hat{\varepsilon}_{j-1}\hat{\varepsilon}_{j}^{t}\right)\right)\right]^{2}.
\end{align*}

Consider
$M_{5}:=\frac{1}{pT^{2}}tr\left(\left(\sum_{i=1}^{T}\hat{\varepsilon}_{i}\hat{\varepsilon}_{i-1}^{t}\right)\left(\sum_{j=1}^{T}\hat{\varepsilon}_{j-1}\hat{\varepsilon}_{j}^{t}\right)\right)$,

\begin{align*}
  \mathbb{E}\left(M_{5}\right) & = \frac{1}{pT^{2}}\sum_{i,j=1}^{p}\mathbb{E}\left(\sum_{t=1}^{T}\hat{\varepsilon}_{it}\hat{\varepsilon}_{j,t-1}\right)^{2}\\
  & =
  \frac{1}{pT^{2}}\sum_{i,j=1}^{p}\sum_{t=1}^{T}\mathbb{E}\left(\hat{\varepsilon}_{it}^{2}\right)\mathbb{E}\left(\hat{\varepsilon}_{j,t-1}^{2}\right)=c\hat{\sigma}_{ij}^{4}.
\end{align*}
\noindent
Moreover,
\begin{align*}
  Var\left(M_{5}\right) & \leq \mathbb{E}\left(\frac{1}{pT^{2}}\sum_{i,j=1}^{p}\left(\sum_{t=1}^{T}\hat{\varepsilon}_{it}\hat{\varepsilon}_{j,t-1}\right)^{2}\right)^{2}\\
  & =  \frac{1}{p^{2}T^{4}}\mathbb{E}\left(\sum_{i,j=1}^{p}\sum_{t=1}^{T}\hat{\varepsilon}_{it}^{2}\hat{\varepsilon}_{j,t-1}^{2}\right)^{2}+\frac{1}{p^{2}T^{4}}\mathbb{E}\left(\sum_{i,j=1}^{p}\sum_{t_{1}\neq t_{2}}\hat{\varepsilon}_{it_{1}}\hat{\varepsilon}_{j,t_{1}-1}\hat{\varepsilon}_{it_{2}}\hat{\varepsilon}_{j,t_{2}-1}\right)^{2}\\
  & = O\left(1\right)+O\left(\frac{1}{T^{2}}\right)=O\left(1\right).
\end{align*}
\noindent
Therefore
$L^{4}\left(F^{\hat{A}},F^{\hat{\sigma}_{ij}^{-4}\hat{A}}\right)\rightarrow0,
a.s.$

% \section*{Acknowledgements}
% And this is an acknowledgements section with a heading that was
% produced by the $\backslash$section* command. Thank you all for
% helping me writing this \LaTeX\ sample file. See \ref{suppA} for the
% supplementary material example.

% \begin{supplement}
%   \sname{Supplement A}\label{suppA} \stitle{Title of the Supplement
%   A}
%   \slink[url]{http://www.e-publications.org/ims/support/dowload/imsart-ims.zip}
%   \sdescription{Dum esset rex in accubitu suo, nardus mea dedit
%   odorem suavitatis. Quoniam confortavit seras portarum tuarum,
%   benedixit filiis tuis in te. Qui posuit fines tuos}
% \end{supplement}

\end{document}